\documentclass[reqno]{amsart}
\usepackage{amsfonts}
\usepackage{amsmath}
\usepackage{amsthm}
\usepackage{amssymb}
\usepackage{mathtools}
\usepackage{mathrsfs}
\usepackage{dsfont}
\usepackage{bbding}
\usepackage{bbm}
\usepackage{upgreek}
\usepackage[latin1,utf8]{inputenc}
\usepackage{fancyhdr}
\usepackage{enumerate}
\usepackage[hang,flushmargin]{footmisc} 

\DeclareFontFamily{U}{dutchcal}{\skewchar\font=45 }
\DeclareFontShape{U}{dutchcal}{m}{n}{<-> s*[1.0] dutchcal-r}{}
\DeclareFontShape{U}{dutchcal}{b}{n}{<-> s*[1.0] dutchcal-b}{}
\DeclareMathAlphabet{\mathlcal}{U}{dutchcal}{m}{n}
\SetMathAlphabet{\mathlcal}{bold}{U}{dutchcal}{b}{n}

\usepackage{graphicx}
\usepackage[section]{placeins}
\usepackage{float}
\usepackage[dvipsnames]{xcolor}

\usepackage[numbers]{natbib}

\usepackage{hyperref}
\usepackage[capitalise]{cleveref}
\hypersetup{
    colorlinks=true,
    linkcolor=Violet,
    citecolor=Violet,
    urlcolor=Violet,
}

\theoremstyle{plain}
\newtheorem{theorem}{Theorem}[section]
\newtheorem{lemma}[theorem]{Lemma}
\newtheorem{proposition}[theorem]{Proposition}
\newtheorem{proposition*}{Proposition}
\newtheorem{corollary}[theorem]{Corollary}

\theoremstyle{definition}

\newtheorem{example}{Example}

\theoremstyle{remark}

\newtheorem*{remark*}{Remark}

\numberwithin{equation}{section}

\usepackage{accents}
\newcommand\thickbar[1]{\accentset{\rule{.4em}{.5pt}}{#1}}

\makeatletter
\@namedef{subjclassname@2020}{%
  \textup{2020} Mathematics Subject Classification}
\makeatother

\title[Infinite-dimensional Gibbs algebras]{Infinite-dimensional genetic and evolution algebras generated by Gibbs measures}

\author[C.F. Coletti \and L.R. de Lima \and D.A. Luiz]{Cristian F. Coletti \and Lucas R. de Lima \and Denis A. Luiz}

\address{Centro de Matem\'atica, Computa\c{c}\~ao e Cogni\c{c}\~ao, Universidade Federal do ABC\\
Av. dos Estados, 5001\\
09210-580 Santo Andr\'e, S\~ao Paulo\\
Brazil.}
\email{cristian.coletti@ufabc.edu.br}
\email{denis.luiz@ufabc.edu.br, denis.aluiz@gmail.com}

\address{Institut of Mathematics and Statitics, University of S\~ao Paulo, Rua do Mat\~ao, 1010, 05508-090 S \~ao Paulo -SP, Brazil.}
\email{lrdelima@ime.usp.br, lrdelimath@gmail.com}



\thanks{{\bf Funding:} This study was financed in part by the Coordenação de Aperfeiçoamento de Pessoal de Nível Superior - Brasil (CAPES) - Finance Code 001. It was also supported by grants \#2017/10555-0 and \#2019/19056-2 S\~ao Paulo Research Foundation (FAPESP)}

\begin{document}

\nocite{*}

\begin{abstract}
    Genetic and evolution algebras arise naturally from applied probability and stochastic processes. Gibbs measures describe interacting systems commonly studied in thermodynamics and statistical mechanics with applications in several fields. Here, we consider that the algebras are determined by configurations of finite spins on a countable set with their associated Gibbs distributions. The model preserves properties of the finite-dimensional Gibbs algebras found in the literature and extend their results. We introduce infertility in the genetic dynamics when the configurations differ macroscopically. It induces a decomposition of the algebra into a direct sum of fertile ideals with genetic realization.
    
    The proposed infinite-dimensional algebras are commutative, non-associative, with uncountable basis and zero divisors. The properties of Gibbs measures allow us to deal with the difficulties arising from the algebraic structure and obtain the results presented in this article.
\end{abstract}

\maketitle


\section{Introduction}

This work has one main purpose consisting in developing a theory of infinite-dimensional genetic (and evolution) algebras which derives from particle configurations corresponding to infinite volume Gibbs measures. We also illustrate the connection between probability measures and algebras. 

On one hand, genetic algebras and evolution algebras are commonly non-associative algebras that have received much attention in the recent past, and it is still a rich and fruitful area of research nowadays. These kind of algebras appeared as a tool developed to understand the evolution laws of genetics. The study of population genetics and its connection with algebras began in 1924 with the work of Bernstein \cite{bernstein1924} and his study about evolution operators. Indeed, in the case when the inheritance can be described by an evolution operator we can associate an algebra to each population with a finite number of genotypes. Etherington \cite{etherington1939} made use of finite-dimensional non-associative algebras to study Mendelian population genetics. In 1981, Holgate \cite{holgate1981} addressed the problem of studying the stochastic population process by relating a non-associative algebra capable of describing the composition of an infinitely large population. Infinite-dimensional genetic  algebras were also further investigated by Holgate \cite{holgate1989} and W\"ortz-Busekros \cite{woertz-busekros1980}. Recently, in 2022, Vidal et al. \cite{vidal2022hilbert,vidal2022hilbert-Hamel} proposed models of Hilbert evolution algebras in a infinite-dimensional setting.

On the other hand,  Gibbs distribution is a central concept in statistical mechanics used to relate microscopic and macroscopic large systems of interacting particle processes. Indeed, infinite-volume Gibbs measures arise naturally, but not trivially, in the study of infinite systems at equilibrium.

Ganikhodzhaev and Rozikov \cite{ganikhodzhaev2006} introduced an evolutionary operator generated by finite-volume Gibbs measures that has an underlying associated genetic algebra. Rozikov and Tian \cite{rozikov2011} proceeded similarly to define finite-dimensional evolution Gibbs algebras.  They also gave guidelines to define what it would possibly be an infinite dimensional algebra determined by Gibbs measures. However, their proposal entails three main issues:
\begin{enumerate}[(i)]
    \item There is no compatibility between the previously defined finite-dimensional algebras and the proposed infinite-dimensional algebras;
    \item There is a boundary effect acting on the $\Lambda_n \uparrow \mathbb{L}$ that affects the existence and/or uniqueness of the limits that determine the structure coefficients of the algebra;
    \item The product is not well-defined even when the limit exists due to the resulting infinite sums.
\end{enumerate}

To address the concerns highlighted above, we preserve the framework they used for finite-volume Gibbs measures introducing infertility into the reproduction dynamics. We allow only configurations that differ by a finite number of sites to produce offspring. This approach is similar to that employed in a model of sex-linked inheritance presented in W\"ortz-Busekros \cite[Sec.~8.B]{woertz-busekros1980} and it reflects the nature of Gibbs measures, which are commonly described by their microscopic components. 

We emphasise that some of the issues presented above are not necessarily a problem for the operators in \cite{ganikhodzhaev2006}, because their main goal is not to work with the associated algebras. In fact, the definition of the algebra that we propose here is not suitable to study their stochastic operators on the whole algebra, but it still possible to operate them when they are defined on the Markov ideals or on other subalgebras.


\subsection{On Gibbs measures} \label{sec:Gibbs}

The theory of Gibbs measures plays a central role in equilibrium statistical mechanics. They are probability measures that describe the collective macroscopic behavior of a system based on information from its interacting microscopic components.
In this section we provide a short introduction to Gibbs measures. For a deeper discussion we refer the reader to Georgii \cite{georgii2011}.

Let $(S,\mathscr{S})$ be a measurable space and regard $S$ as the spin space of a countable set $\mathbb{L}$. From now on we assume $S$ to be finite and $\mathscr{S} = 2^S$, the power set of $S$. The following definitions and results rely on the finiteness of $S$ and the counting measure $\lambda$ on $(S,\mathscr{S})$. A configuration $\sigma$ is an element of $\Omega = {S}^\mathbb{L}$. Write $(\Omega, \mathscr{F})$ for the product measurable space with $\mathscr{F}= \mathscr{S}^{\otimes \mathbb{L}}$. Set $\mathscr{F}_L = \mathscr{S}^{\otimes L}$ for $L \subseteq \mathbb{L}$ and let $\mathcal{L}$ be the set of all finite subsets of $\mathbb{L}$.

The \emph{interaction potential} $\Phi$ (or \emph{potential} for short) is a family of functions $\Phi=(\Phi_A)_{A\in \mathcal{L}}$ such that $\Phi_A: \Omega \to \mathbb{R}$ is $\mathscr{F}_A$-measurable and, for all $\Lambda \in \mathcal{L}$ and all $\sigma\in \Omega$, the \emph{total energy} of $\sigma$ in $\Lambda$ for $\Phi$
\begin{equation*}
    H_\Lambda^\Phi(\sigma) := \sum_{\substack{A \in \mathcal{L}, ~A \cap \Lambda \neq \emptyset}}\Phi_A(\sigma)
\end{equation*}
exists. We call $H_\Lambda^\Phi$ the \emph{Hamiltonian} in $\Lambda$ for $\Phi$. The potential is \textit{admissible} when, for all $\sigma \in \Omega$ and all $\Lambda \in \mathcal{L}$, $H_\Lambda^\Phi(\sigma)$ is finite and converges unconditionally, meaning it does not depend on the ordering of $\mathcal{L}$. The Boltzmann-Gibbs distributions are usually written in terms of $\exp(-\beta H)/Z$ on the phase space. In that case, $\beta$ is the inverse temperature $1/(kT)$ and $Z$ is a partition function. Here we let the Hamiltonian absorb $\beta$ by setting $\beta=1$.

Let $\sigma_L \in S^{L}$ denote the restriction of $\sigma$ to $L \subseteq \mathbb{L}$. Consider $\zeta \in S^L$ and $\xi \in S^{L'}$ with $L \cap L' = \emptyset$, then we write $\zeta \xi$ for the configuration in $S^{L \dot\cup L'}$  such that $(\zeta\xi)(x) = \zeta(x)$ when $x \in L$ and $(\zeta\xi)(x) = \xi(x)$ if $x \in L'$.

Denote by $h_\Lambda^\Phi(\sigma) := \exp(-H_\Lambda^\Phi(\sigma))$ the Boltzmann factor. Consider $\Lambda^c$ as $\mathbb{L} \setminus \Lambda$ for all $\Lambda\in\mathcal{L}$. Define the partition function $Z_\Lambda^\Phi$ for $\Lambda$ and $\Phi$ on $\Omega$ to be given by
\begin{equation*}
    Z_\Lambda^\Phi (\sigma):= \sum_{\zeta \in S^\Lambda} h_\Lambda^\Phi(\zeta\sigma_{\Lambda^c}).
\end{equation*}

Let $\eta \in \Omega$ and  $\Lambda \in \mathcal{L}$. Consider $\Phi$ an admissible potential. Denote by $\mathds{1}$ the indicator function. Then the measure $\gamma_\Lambda^\Phi(\cdot\mid\eta)$ on $(\Omega,\mathscr{F})$ defined by
\begin{equation}\label{eq:DLR}
    \gamma_\Lambda^\Phi (\mathsf{E} \mid \eta) :=\frac{1}{Z_\Lambda^\Phi(\eta)}\sum_{\zeta \in S^\Lambda} h_\Lambda^\Phi(\zeta\eta_{\Lambda^c}) \mathds{1}_\mathsf{E}(\zeta\eta_{\Lambda^c})  \quad \text{for all } \mathsf{E} \in \mathscr{F} 
\end{equation}
is called the \emph{Gibbs distribution} with boundary condition $\eta_{\Lambda^c}$, and an admissible $\Phi$.

Set $\mathscr{T}_\Lambda := \mathscr{F}_{\mathbb{L}\setminus\Lambda}$to be the external $\sigma$-algebra of $\Lambda$. The \emph{tail $\sigma$-algebra} is given by 
\[\mathscr{T} := \bigcap_{\Lambda \in \mathcal{L}} \mathscr{T}_{\Lambda}.\]

Recall that if $\nu$ is a probability measure on $(\Omega,\mathscr{F})$ and $\mathscr{H} \subseteq \mathscr{F}$ is a sub-$\sigma$-algebra, then $\nu(\mathsf{E}\mid\mathscr{H}) := \mathbbm{E}_\nu[\mathds{1}_\mathsf{E}\mid\mathscr{H}]$.

A \emph{Gibbs measure} $\mu$ on $(\Omega,\mathscr{F})$ for an admissible $\Phi$ is a random field such that
\begin{equation*}
\mu(\mathsf{E}\mid\mathscr{T}_\Lambda) = \gamma_\Lambda^\Phi(\mathsf{E}\mid\cdot) \quad\mu\text{-a.s.} \quad\text{for all }\mathsf{E} \in \mathscr{F}\text{ and }\Lambda \in \mathcal{L}.
\end{equation*}

The set of Gibbs measures for $\Phi$ is denoted by $\mathscr{G}(\Phi)$. The set $\mathscr{G}(\Phi)$ is non-empty  and compact for any admissible $\Phi$ (see Theorem (4.23) of \cite{georgii2011} with $\lambda$ the counting measure on $(S, 2^S)$). An interaction potential $\Phi$ is said to exhibit a \emph{phase transition} when $\vert\mathscr{G}(\Phi)\vert>1$. Therefore a potential $\Phi$ may be associated with multiple Gibbs measures that locally share the same Gibbs distribution. It is also possible that a Gibbs measure $\mu$ on $(\Omega,\mathscr{F})$ is such that $\mu \in \mathscr{G}(\Phi) \cap \mathscr{G}(\Psi)$ for two distinct potentials $\Phi$ and $\Psi$. These properties highlight the need to explicitly declare both the measure and the associated potential in our further constructions.

\begin{example}[The Potts' model without external field] \label{ex:potts}
    The Potts model is a classic and relatively simple model which has several applications in different areas (see \cite{rozikov2022}, for instance). Let $\mathcal{G}= (\mathbb{L},E)$ be a countable graph and consider $S=\{1,2, \dots, q\}$ with $q \ge 2$. The interaction potential $\Phi$ for the \textit{Potts' model} with inverse temperature $\beta \geq 0$ is given by
    \[\Phi_A(\sigma)= \left\{\begin{array}{rl}
      -\beta ~\mathds{1}_{\{0\}}\big(\sigma(x)-\sigma(y)\big), & A= \{x,y\} \in E, x \neq y;   \\
      0,   & \text{otherwise.} 
    \end{array}\right.\]
    This is a generalization of the Ising ferromagnetic model which corresponds to the case where $q=2$ and $S=\{1,2\}$ is identified with $\{-,+\}$.
    
    Consider $\mathcal{G}$ to be the square lattice by setting $\mathbb{L}= \mathbb{Z}^2$ and $E=\big\{\{x,y\} \in \mathbb{Z}^2: \|x-y\|=1\big\}$. It is a well known fact that there exists $\beta_c>0$ such that $\vert\mathscr{G}(\Phi)\vert=1$ when $\beta<\beta_c$ and $\vert \mathscr{G}(\Phi)\vert >1$ when $\beta\geq\beta_c$ on $\mathcal{G}$ (see \cite[Sec. ~18.3.5; p.485]{georgii2011} and references therein). \hfill $\blacktriangleleft$ 
\end{example}

To avoid a cumbersome notation, let us denote, for all $\eta,\xi \in \Omega$ and all $\Lambda \in \mathcal{L}$, 
\[\mu_\Lambda(\mathsf{E}\mid \xi):= \mu(\mathsf{E}\mid \xi_{\Lambda^c})= \gamma_\Lambda^\Phi(\mathsf{E}\mid \xi) \quad \text{and} \quad \mu_\Lambda(\eta\mid \xi):=\mu_\Lambda(\{\eta\}\mid \xi).\]

An immediate consequence of the finiteness of $S$ is that a Gibbs measure $\mu$ with boundary conditions assigns strictly positive probability to each configuration in the cylinder set. In other words, if $\eta \in \mathcal{C}_\Lambda[\xi]:= \{\sigma \in \Omega : \sigma_{\Lambda^c} = \xi_{\Lambda^c}\}$, then $\mu_{\Lambda}(\eta\mid \xi) >0$.  The correspondence established by the Hammersley–Clifford, Sullivan, and Kozlov theorems between Gibbs measures and Markov random fields has even stronger consequences, ensuring that only Gibbs measures possess this property (see \cite{kozlov1974,kozlov1977,sullivan1973}). This is a central fact that we use to define the genetic and evolution algebras.

In the case that $\mathbb{L}$ is finite it follows that $\mu = \mu_{\mathbb{L}}$ since $\gamma_{\mathbb{L}}^\Phi(\mathsf{E}\mid \cdot)$ is constant for each fixed $\mathsf{E} \in \mathscr{F}$. Therefore,\[\mu(\mathsf{E}) = \frac{1}{Z_\mathbb{L}^\Phi}\sum_{\zeta \in \mathsf{E}} h_{\mathbb{L}}^\Phi(\zeta)  \quad \text{for all } \mathsf{E} \in \mathscr{F} \text{ when }\vert\mathbb{L}\vert <+\infty. \]
One can easily see that a simple consequence of $\mathbb{L}$ being finite is that $\vert\mathscr{G}(\Phi)\vert =1$.

Two potentials $\Phi$ and $\Psi$ are \emph{equivalent} if, for all $\Lambda \in \mathcal{L}$, the Hamiltonian $H_\Lambda^{\Phi - \Psi}$ is $\mathscr{T}_\Lambda$-measurable on $\Omega$. Let $\Phi \sim \Psi$ stand for the equivalence between the potentials $\Phi$ and $\Psi$. In particular, if $\Phi \sim \Psi$, then, for all $\Lambda \in \mathcal{L}$
\begin{equation} \label{eq:equiv.potentials}
    h_\Lambda^\Phi = h_\Lambda^{\Phi-\Psi}h_\Lambda^{\Psi}\quad \text{and} \quad Z_\Lambda^\Phi = h_\Lambda^{\Phi-\Psi}Z_\Lambda^{\Psi}.
\end{equation}
We will later explore the effect of this equivalence relation on genetic and evolution algebras generated by measures in $\mathscr{G}(\Phi)$ and $\mathscr{G}(\Psi)$.

\begin{example}[Equivalent interaction potentials]
    Consider $\{c_A\}_{A\in\mathcal{L}}$ a sequence of constants such that $\sum_{A \in \mathcal{L}} \vert c_A\vert < +\infty$ unconditionally. Let $\Phi$ be an admissible interaction potential on $\Omega=S^\mathbb{L}$. Define $\Psi$ to be given by
    \[\Psi_A:= \Phi_A + c_A.\]

    Then $\Psi$ is admissible since $\vert H_\Lambda^\Psi\vert  \leq \vert H_\Lambda^\Phi\vert  + \sum_{A \in \mathcal{L}} \vert c_A\vert  <+\infty$ for all $\Lambda \in \mathcal{L}$. Note that $H_\Lambda^{\Phi-\Psi} \equiv \sum_{A \cap \Lambda \neq \emptyset} c_A$ is constant, which implies that $H_\Lambda^{\Phi-\Psi}$ is $\mathscr{T}_\Lambda$-measurable. Therefore, $\Phi \sim \Psi$. \hfill $\blacktriangleleft$ 
\end{example}

Most of models that define a Gibbs measure are associated with graphs. In fact, it is common to consider that the interaction potential is intrinsically related to a graph structure. Here, we define a notation that resembles the structure of a graph while still keeps the potential generic. Let the \textit{interaction potential boundary} of a set $L \subseteq \mathbb{L}$ be given by
\[\partial_\Phi L := \bigcup_{\substack{A \in \mathcal{L}, ~A \cap L \neq \emptyset\\\Phi_A \not \equiv 0}} A \setminus L,\]
and the \textit{interaction potential closure} to be the set $cl_\Phi (L) := L ~\dot\cup~ \partial_\Phi L$. Let the $\Phi$\textit{-neighbourhood} of $x$ be defined by \[\mathbf{n}_x^\Phi := \big\{ \{x,y\}: y\in\partial_\Phi \{x\} \big\}.\] A potential $\Phi$ has \emph{finite range} if, for all $x \in \mathbb{L}$, $\mathbf{n}_x^\Phi \in \mathcal{L}$. Note that now the potential actually determines a graph $\mathcal{G}_\Phi =(\mathbb{L},E_\Phi)$ with $E_\Phi= \big\{\{x,y\}: x \in \mathbb{L}, y \in \mathbf{n}_x^\Phi\big\}$. Therefore $x\sim y$ in $\mathcal{G}_\Phi$ indicates that it is possible for the spin at $x$ interact with the spin at $y$ to assign a Gibbs measure $\mu \in \mathscr{G}(\Phi)$ to the configurations of $\Omega$. 

One of the central elements to define the algebras in the next sections will be $\mathfrak{D}_{\sigma\eta}$ that stands for the set of discrepancies of a pair $\sigma, \eta \in \Omega$, which is given by \[\mathfrak{D}_{\sigma\eta} := \{x \in \mathbb{L}:\sigma(x) \neq \eta(x)\}.\]
We say that $\sigma$ and $\eta$ differ microscopically when $\mathfrak{D}_{\sigma\eta} \in \mathcal{L}$; otherwise they differ macroscopically.

The basic algebraic concepts that determines the structures studied in this article are presented in the following subsections.


\subsection{On genetic algebras} \label{sec:genetic.algebras}

There is a wide variety of algebras in mathematical genetics that characterize reproduction
laws and genetic changes in a population. The genetic algebras define a large class of algebraic structures that describes genetic inheritance in successive generations. For a deeper discussion on algebras in genetics, we refer the reader to W\"orz-Busekros \cite{woertz-busekros1980} or Rozikov \cite{rozikov2020}.  

Here we focus on \textit{gametic algebras} which will be simply referred as \textit{genetic algebras}. We will also allow some abstract generalizations preserving the basic biological motivation. Let the genetic algebra $\mathlcal{A}$ be defined as a $\mathbb{K}$-module generated by the Hamel basis $\{e_i\}_{i\in I}$ with product given by bilinear extension of
\[e_i \cdot e_j = \sum_{k \in I} a_{ij,k}e_k.\]
We consider in this article $\mathbb{K}=\mathbb{R}$ or $\mathbb{K}=\mathbb{C}$ with $a_{ij,k} \in \mathbb{R}$ for all $i,j \in I$. Furthermore, the fertility rates are given by $\sum_{k \in I} a_{ij,k} \in [0,1]$. In case of symmetric inheritance, $a_{ij,k}=a_{ji,k}$ for all $i,j \in I$, which implies $\mathlcal{A}$ to be commutative but not necessarily associative. An algebra is called \textit{Markov} (or with \textit{genetic realization}) when the fertility rates $\sum_{k \in I} a_{ij,k} =1$ for all $i,j \in I$. Therefore, the structure coefficients $a_{ij,k}$ serve as \textit{transition probabilities} of a Markov operator when $\mathlcal{A}$ is Markov, as follows. Consider a Markov chain associated with each $e_i$. The probability of $e_i$ producing $e_k$ in the next time step (or generation) when it couples with $e_j$ is $P_i(k\mid j)= a_{ij,k}$. This is the motivation for calling them Markov algebras and setting the structure coefficients to be real numbers, extending the concept to more general genetic algebras.

The biological interpretation may regard $\{e_i\}_{i \in I}$ as a set of genetically distinct gametes. A union of two gametes $e_i$ and $e_j$ forms a zygote $e_ie_j$ that produces $\widetilde{a}_{ij,k}$ gametes $e_k$ which survive the next generation. Let $m_{ij}:= \sum_{k \in I}\widetilde{a}_{ij,k}$ and fix $a_{ij,k} = \widetilde{a}_{ij,k}/m_{ij}$ when $m_{ij} \in (0,+\infty)$. Define $a_{ij,k}=0$ for all $k \in I$ in the case that $m_{ij}=0$. Then $a_{ij,k}$ stands for the probability that a gamete $e_l$ is equal to $e_k$ given that it was produced by the zygote $e_ie_j$. 

Considering this biological interpretation above, the fertility rates are either $0$ or $1$. Hence, we may allow infertility of a $e_ie_j$. It is worth to point out that this interpretation could be changed (see \cite[Sec. 1.A.]{woertz-busekros1980}) to exhibit fertility rates in the interval $[0,1]$. However, later in the text the structure coefficients will be defined by probability measures and not by the number of descendants. This brings us to the interpretation presented above.

Let $\operatorname{hom}(\mathlcal{A}',\mathlcal{A}'')$ denote the set of \textit{algebra homomorphisms} from $\mathlcal{A}'$ to $\mathlcal{A}''$, \textit{i.e.}, the set of linear maps $\varphi$ such that $\varphi(u\cdot v)=\varphi(u)\cdot \varphi(v)$. If the $\mathbb{K}$-algebra $\mathlcal{A}$ admits a non-zero $w \in \operatorname{hom}(\mathlcal{A},\mathbb{K})$, the ordered pair $(\mathlcal{A},w)$ is called a \textit{weighted algebra} and $w$ a \textit{weight homomorphism}. Weighted algebras are essential to define other types of genetic algebras (e.g., baric, Bernstein, and train algebras) and explore their algebraic properties.

\subsection*{Finite-dimensional Gibbs genetic algebras} \label{sec:fin.dim.Gibbs.gen}

The following construction is adapted from Ganikhodzhaev and Rozikov \cite{ganikhodzhaev2006}. They studied evolutionary stochastic operators determined by Gibbs measures on finite lattices. The underlining algebra is described below and it corresponds to the evolutionary state of free population through generations.

Consider $\mathcal{G}= (\mathbb{L},E)$ a graph with $\mathbb{L}$ finite and $S$ the finite spin set. We proceed with the construction of the finite Gibbs measures given in \cref{sec:Gibbs}. Let $C(x) \subseteq \mathbb{L}$ denote the connected component of $x$ in $\mathcal{G}$. Define $\mathscr{C}= \{C(x): x\in\mathbb{L}\}$ the set of connected components of $\mathcal{G}$. Then  $\mathscr{C}$ is a partition of $\mathbb{L}$ and we will later call $\Delta \in \mathscr{C}$ a \textit{cluster}.

The next generation of a pair of configurations $\sigma,\eta \in \Omega$ is given by the configurations that coincide with either $\sigma$ or $\eta$ on each cluster $\Delta \in \mathscr{C}$. 
Hence, the offspring produced by $\sigma,\eta \in \Omega$ is given by the set
\[\thickbar{\Omega}_{\sigma\eta}=\thickbar{\Omega}_{\sigma\eta}(\mathscr{C}) := \big\{ 
\zeta \in \Omega : \forall \Delta \in \mathscr{C}\big(\zeta_\Delta = \sigma_\Delta \text{ or } \zeta_\Delta = \eta_\Delta \big) \big\}.\]

Let us now define $\mathlcal{A}_f(\mathscr{C},\mu, \Omega)$ as the finite-dimensional algebra generated by the Gibbs measure $\mu$ to be free $\mathbb{K}$-module generated by $\{e_\sigma\}_{\sigma\in \Omega}$ and product given by bilinear extension of
\[e_\zeta \cdot e_\eta = \sum_{\sigma \in \thickbar{\Omega}_{\zeta\eta}} a_{\zeta\eta,\sigma}e_\sigma \quad\text{with}\quad a_{\zeta\eta,\sigma} = \frac{\mu(\sigma)}{\mu(\thickbar{\Omega}_{\zeta\eta})}.\]


\subsection{On evolution algebras}

This class of algebras was first introduced to study self-reproduction of alleles in non-Mendelian genetics.  Its particular structure made it interesting for many other
mathematical fields, such as graph theory, stochastic processes,  and mathematical physics. The best general reference here is Tian \cite{tian2018}. Let the evolution algebra $\mathlcal{E}$ be defined as the $\mathbb{K}$-module generated by $\{e_i\}_{i\in I}$ with product given by bilinear extension of 
\[e_i \cdot e_j = \left\{\begin{array}{cc}
    \sum_{k \in I} b_{ik}e_k, & \text{if } i=j \\
    0 ,& \text{otherwise.}
\end{array}\right.\]

Observe that they are a specific type of the genetic algebras defined in \cref{sec:genetic.algebras}. This can be verified by setting $a_{ij,k} = \delta_{ij}b_{ik}$. 

Their algebraic structure makes them particularly interesting and they stand out among other genetic algebras. For instance, they may exhibit algebraic persistency, algebraic transiency, and algebraic periodicity. Note that the evolution algebras are commutative and it suffices to exist $e_i \neq e_j$ such that $e_i^2 \cdot e_j \neq 0$ for the algebra to be non-associative.

\subsection*{Finite-dimensional Gibbs evolution algebras} 

In a later study, ozikov and
Tian \cite{rozikov2011} defined finite dimensional evolution algebras generated by Gibbs measures. The algebras are similar to the genetic algebras defined above, but now the ordered pairs of configurations in $\Omega^2$ define its generating set.

The offspring now is given by each $(\sigma,\eta) \in \Omega^2$ producing $\thickbar{\Omega}_{\sigma\eta}^2$ (recall the notation defined in \cref{sec:fin.dim.Gibbs.gen}). We write $\mu^{\otimes 2}:= \mu \otimes\mu$ for the product measure on $\Omega^2$. Let $\mathlcal{E}_f(\mathscr{C},\mu,\Omega)$ stand for the finite-dimensional algebra generated by Gibbs measure $\mu$ defined as the free $\mathbb{K}$-module generated by $\{e_{\sigma\eta}\}_{(\sigma,\eta)\in \Omega^2}$ with product given by bilinear extension of
\[e_{\zeta\xi} \cdot e_{\zeta\xi} = \sum_{(\eta,\sigma) \in \thickbar{\Omega}_{\zeta\xi}^2} b_{\zeta\xi,\eta\sigma}e_{\eta\sigma} \quad\text{with}\quad b_{\zeta\xi,\eta\sigma} = \frac{\mu(\eta)\mu(\sigma)}{\mu^{\otimes2}\left(\thickbar{\Omega}_{\zeta\eta}^2\right)},\]
and $e_{\zeta\xi} \cdot e_{\zeta'\xi'} = 0$ when $(\zeta,\xi) \neq (\zeta',\xi')$.


 \section{Infinite-dimensional Gibbs genetic algebras} \label{sec:infinite.dim.genetic.algebras}

The interacting microscopic components from the Gibbs formalism will be used to define a genetic dynamics its configuration space. Consider that the offspring is determined by replicating the configuration of one of the parents in the given clusters (also \textit{loci set} or \textit{connected component}) of the system. The configurations that differ macroscopically do not produce offspring. Thus only configurations with finite discrepancy sets are allowed to reproduce, \textit{i.e.}, $\sigma, \eta\in \Omega$ such that $\mathfrak{D}_{\sigma\eta} \in \mathcal{L}$ (recall notation from \cref{sec:Gibbs}).

From now on, we define the configuration space $\Omega = S^\mathbb{L}$ with  $S$ finite and $\mathbb{L}$ countable. Let $\mathscr{C}$ be a partition of $\mathbb{L}$ so that each $\Delta \in \mathscr{C}$ is called a cluster. The set of the offspring produced by a pair $\sigma, \eta \in \Omega$ is defined by
\[\Omega_{\sigma\eta} = \Omega_{\sigma\eta}(\mathscr{C}) := \big\{\zeta \in \Omega : \forall \Delta \in \mathscr{C}\big(\zeta_\Delta = \sigma_\Delta \text{ or } \zeta_\Delta = \eta_\Delta \big) \text{ and }\mathfrak{D}_{\sigma\eta} \in \mathcal{L}\big\}.\] 
\begin{figure}[!htb]
    \centering
    \includegraphics[width=280pt]{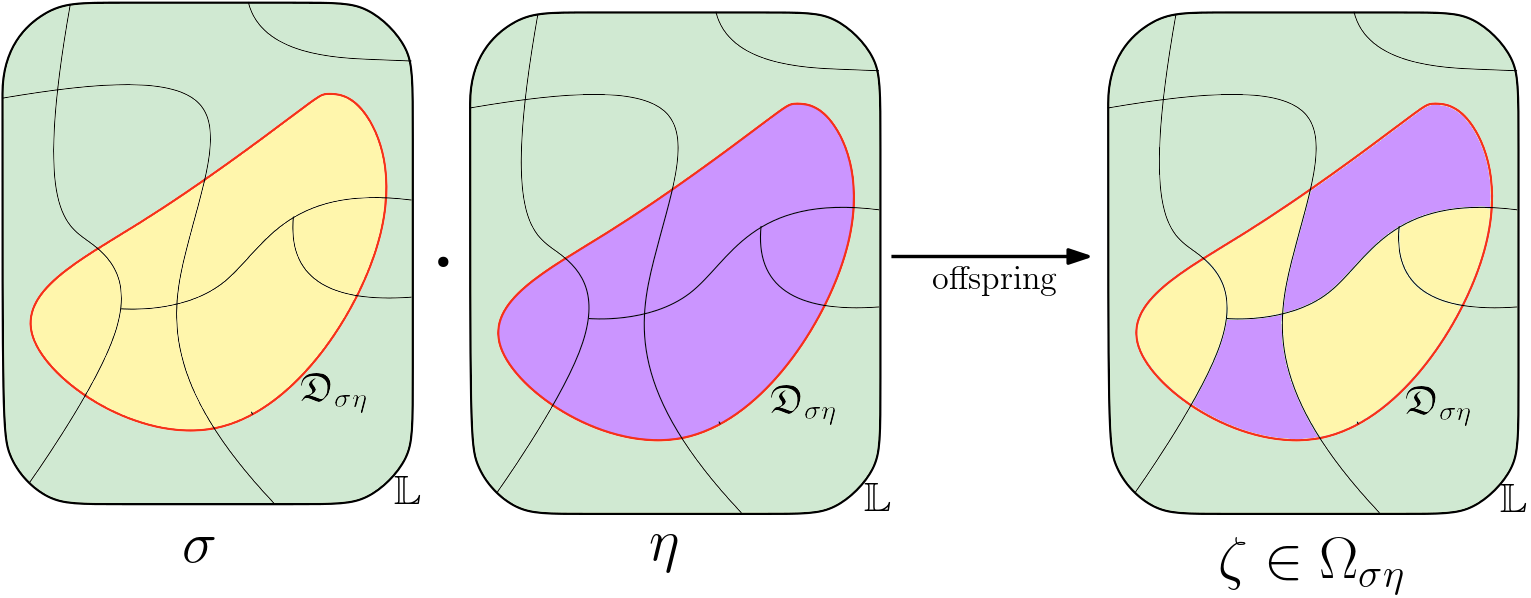}
    \caption{Offspring alternating configurations in the clusters intersecting $\mathfrak{D}_{\sigma\eta}$ (the set of parental discrepancies).}
    \label{fig:offspring}
\end{figure}

Let $\langle \mathfrak{B} \rangle$ denote the $\mathbb{K}$-module spanned by $\mathfrak{B}$ as a Hamel (algebraic) basis. Consider $\mathfrak{B}_{I}:= \{e_i\}_{i \in I}$ to be a basis indexed by a set $I$ and fix $\mathbb{K}$ to be equal to $\mathbb{R}$ or $\mathbb{C}$. The $\mathscr{C}$-genetic Gibbs algebra generated by $\mu \in \mathscr{G}(\Phi)$ on $(\Omega,\mathscr{F})$ is the free $\mathbb{K}$-module $\mathcal{A}(\mathscr{C},\mu,\Phi,\Omega)=\left\langle \mathfrak{B}_{\Omega} \right\rangle$ with product given by bilinear extension of
\[e_\zeta \cdot e_\eta = \sum_{\sigma \in \Omega_{\zeta\eta}} c_{\zeta \eta, \sigma} e_\sigma\] where
\[c_{\zeta \eta, \sigma} = \frac{\mu(\sigma\mid \zeta_{(\mathfrak{D}_{\zeta\eta})^c})}{\mu(\Omega_{\zeta\eta}\mid \zeta_{(\mathfrak{D}_{\zeta\eta})^c})} = \frac{\gamma_{\mathfrak{D}_{\zeta\eta}}^\Phi(\sigma\mid \zeta)}{\gamma_{\mathfrak{D}_{\zeta\eta}}^\Phi(\Omega_{\zeta\eta}\mid \zeta)} = \frac{h_{\mathfrak{D_{\zeta \eta}}}^\Phi(\sigma)}{\sum_{\xi \in \Omega_{\zeta\eta}}h_{\mathfrak{D_{\zeta \eta}}}^\Phi(\xi)}\]
when $\sigma \in \Omega_{\zeta\eta}$. We may write $c_{\zeta \eta, \sigma}=0$ when $\sigma \not\in \Omega_{\zeta\eta}$. Define for all $\eta \in \Omega$
\[\mathtt{E}^\eta := \big\{\sigma \in \Omega : \vert\mathfrak{D}_{\sigma\eta}\vert< +\infty \big\}, \text{ and} \quad \mathtt{F}^\eta:=\left\langle \mathfrak{B}_{\mathtt{E}^\eta} \right\rangle.\]
We call $\mathtt{E}^\eta$ the $\eta$\textit{-fertile class} and $\mathtt{F}^\eta$ the $\eta$\textit{-fertile ideal} (see \cref{thm:genetic.decomposition} below). Note that each class $\mathtt{E}^{\eta}$ is associated with the cylindrical tail of the configurations. It can be translated into Gibbs measures language that the elements of $\eta$-fertile class differ microscopically but are similar macroscopically (\textit{i.e.}, they share the same tail).

It is straightforward to see that if $\eta \in \mathtt{E}^\sigma$, then $\sigma \in \mathtt{E}^\eta$. Thus $\eta \in \mathtt{E}^\eta$ for all $\sigma \in \Omega$. Moreover, given $\sigma, \eta \in \Omega$, either $\mathtt{E}^\sigma = \mathtt{E}^\eta$ or $\mathtt{E}^\sigma \cap \mathtt{E}^\eta = \emptyset$. Therefore $\mathlcal{P}_{\Omega}:=\{\mathtt{E}^\sigma : \sigma \in \Omega\}$ is a partition of $\Omega$. In what follows, we denote by $\mathcal{A} \cong \mathcal{A}'$ the isomorphism of algebras, \textit{i.e.}, when there exists a bijective multiplicative linear map $\varphi \in \operatorname{hom}(\mathcal{A},\mathcal{A}')$.

The following lemma shows that our definition is compatible with the finite Gibbs algebras found in the literature (as defined above in \cref{sec:fin.dim.Gibbs.gen}). An immediate consequence is that the theorems obtained in this article extend to all countable $\mathbb{L}$, not necessarily infinite.

\begin{lemma} \label{lm:finite.alg.equiv}
    Let $\mathbb{L}$ and $S$ be finite. Consider $\Phi$ an admissible interaction potential and $\mu \in \mathscr{G}(\Phi)$ a Gibbs measure. Then, for every fixed set of clusters $\mathscr{C}$,
    \[\mathcal{A}(\mathscr{C},\mu,\Phi,\Omega) \cong \mathlcal{A}_f(\mathscr{C},\mu,\Omega).\]
\end{lemma}
\begin{proof}
    It suffices to show that the two definitions are equivalent. First, observe that, since $\mathbb{L}$ is finite, $\mathcal{L}=2^{\mathbb{L}}$ and $\mathfrak{B}_{\Omega}= \{e_\sigma\}_{\sigma \in \Omega}$. Thus $\Omega_{\xi\eta}=\thickbar{\Omega}_{\xi\eta}$ for all $\xi,\eta \in \Omega$ and $\mathscr{C}$ fixed. It then suffices to verify that all the structure coefficients $a_{\zeta\eta,\sigma}$ and $c_{\zeta\eta,\sigma}$ are equal, \textit{i.e.}, for all $\zeta,\eta \in \Omega$ and $\sigma \in \Omega_{\zeta\eta}$, 
    \[a_{\zeta\eta,\sigma}=\frac{\mu(\sigma)}{\mu(\Omega_{\zeta\eta})} =\frac{\mu(\sigma\mid \zeta_{(\mathfrak{D}_{\zeta\eta})^c})}{\mu(\Omega_{\zeta\eta}\mid \zeta_{(\mathfrak{D}_{\zeta\eta})^c})}=c_{\zeta\eta,\sigma}.
    \]
    
    We rewrite the equality above as
    \begin{equation} \label{eq:struct.coeff.finite}
        \frac{\mu_{\mathbb{L}}(\sigma)}{\mu_{\mathbb{L}}(\Omega_{\zeta\eta})} = \frac{h_{\mathbb{L}}^\Phi(\sigma)}{\sum_{\xi \in \Omega_{\zeta\eta}}h_{\mathbb{L}}^\Phi(\xi)}= \frac{h_{\mathfrak{D_{\zeta \eta}}}^\Phi(\sigma)}{\sum_{\xi \in \Omega_{\zeta\eta}}h_{\mathfrak{D_{\zeta \eta}}}^\Phi(\xi)} = \frac{\mu_{\mathfrak{D}_{\zeta\eta}}(\sigma\mid \zeta)}{\mu_{\mathfrak{D}_{\zeta\eta}}(\Omega_{\zeta\eta}\mid \zeta)}.
    \end{equation}

    Observe that, for all $\omega \in \Omega$,
    \begin{align*}
        H_{\mathbb{L}}^\Phi(\omega) = \sum_{A \in\mathcal{L} }\Phi_A(\omega) &= \sum_{\substack{A \in \mathcal{L}, ~A \cap \mathfrak{D}_{\zeta\eta} \neq \emptyset}}\Phi_A(\omega) ~~~~+  \sum_{\substack{A \in \mathcal{L}, ~A \cap \mathfrak{D}_{\zeta\eta} = \emptyset}}\Phi_A(\omega)\\
         &= H_{\mathfrak{D}_{\zeta\eta}}^\Phi(\omega) + \sum_{\substack{A \in \mathcal{L}, ~A \cap \mathfrak{D}_{\zeta\eta} = \emptyset}}\Phi_A(\omega).
    \end{align*}
    Let us define
    \[H_{(\zeta\eta)}^\Phi :=\sum_{\substack{A \in \mathcal{L}, ~A \cap \mathfrak{D}_{\zeta\eta} = \emptyset}}\Phi_A(\zeta) \quad \text{and} \quad h_{(\zeta\eta)}^\Phi := \exp\left(-H_{(\zeta\eta)}^\Phi\right).\]

    Recall that $\Phi_A$ is $\mathscr{F}_A$-measurable and, therefore, $\Phi_A(\zeta) = \Phi_A(\omega_{A^c}\zeta_A)$ for all $\omega\in \Omega$. Since $\xi_A \equiv \zeta_A$ for $\xi \in \Omega_{\zeta\eta}$ when $A \cap \mathfrak{D}_{\zeta\eta} = \emptyset$, one has that \[H_{\mathbb{L}}^\Phi(\xi) = H_{\mathfrak{D}_{\zeta\eta}}^\Phi(\xi) + H_{(\zeta\eta)}^\Phi \; \text{ and } \;  h_{\mathbb{L}}^\Phi (\xi)= h_{(\zeta\eta)}^\Phi h_{\mathfrak{D}_{\zeta\eta}}^\Phi(\xi)\quad\text{ for all } \xi \in \Omega_{\zeta\eta}.\] Hence, if $\sigma \in \Omega_{\zeta\eta}$, then
    \[\frac{h_{\mathbb{L}}^\Phi(\sigma)}{\sum_{\xi \in \Omega_{\zeta\eta}}h_{\mathbb{L}}^\Phi(\xi)}= \frac{h_{(\zeta\eta)}^\Phi h_{\mathfrak{D_{\zeta \eta}}}^\Phi(\sigma)}{h_{(\zeta\eta)}^\Phi\sum_{\xi \in \Omega_{\zeta\eta}}h_{\mathfrak{D_{\zeta \eta}}}^\Phi(\xi)}=
    \frac{h_{\mathfrak{D_{\zeta \eta}}}^\Phi(\sigma)}{\sum_{\xi \in \Omega_{\zeta\eta}}h_{\mathfrak{D_{\zeta \eta}}}^\Phi(\xi)},\]
    which proves \eqref{eq:struct.coeff.finite}, and the proof is complete.
\end{proof}

We proceed with some basic facts regarding the discrepancies of two configurations and the offspring produced by a pair of parents in $\Omega$.

\begin{lemma} \label{lm:discrep}
    For all $\sigma,\eta \in \Omega$,  one has that $\vert\Omega_{\sigma\eta}\vert \leq 2^{\vert\mathfrak{D}_{\sigma \eta}\vert}$ when $\mathfrak{D}_{\sigma \eta} \in \mathcal{L}$, and 
    $\vert\Omega_{\sigma\eta}\vert =0$ otherwise. Moreover, for all $\zeta,\xi \in \Omega_{\sigma\eta}$, \[\mathfrak{D}_{\zeta\xi} \subseteq \mathfrak{D}_{\sigma\eta};\] and $\mathfrak{D}_{\sigma\zeta} = \mathfrak{D}_{\sigma\eta}$ if, and only if, $\zeta=\eta$.
\end{lemma}
\begin{proof}
    By the definition of $\Omega_{\sigma\eta}$ , the case $\mathfrak{D}_{\sigma\eta} \not\in \mathcal{L}$ is trivial. Consider now $\mathfrak{D}_{\sigma\eta}$ finite and set \[D_{\sigma\eta} := \{\Delta \in \mathscr{C}: \Delta \cap \mathfrak{D}_{\sigma\eta} \neq \emptyset\}.\] 
    Since $\mathscr{C}$ is a partition of $\mathbb{L}$, there is an unique $\Delta \in D_{\sigma\eta}$ for each element of $\mathfrak{D}_{\sigma\eta}$. Then one can easily see that $\vert D_{\sigma\eta}\vert \leq \vert \mathfrak{D}_{\sigma\eta}\vert $ and all $\zeta \in \Omega_{\sigma\eta}$ is such that
    \[\zeta = \sigma_{L^c} \eta_{L}\]
    with $L = \Delta_1 \Dot{\cup} \cdots\Dot{\cup}\Delta_j$ such that  $\Delta_1 ,\dots,\Delta_j \in D_{\sigma\eta}$. Therefore, there is a bijection between  $\Omega_{\sigma\eta}$ and $2^{D_{\sigma \eta}}$, which implies that 
    \[\vert \Omega_{\sigma\eta}\vert  = 2^{\vert D_{\sigma \eta}\vert } \leq 2^{\vert \mathfrak{D}_{\sigma \eta}\vert }.\]

    Note that, for all $\zeta, \xi \in \Omega_{\sigma\eta}$, one can write $\zeta = \sigma_{\Lambda^c}\eta_{\Lambda}$ and $\xi = \sigma_{(\Lambda')^c}\eta_{\Lambda'}$ with $\Lambda,\Lambda' \subseteq \mathfrak{D}_{\sigma\eta}$. Then $\mathfrak{D}_{\sigma\zeta} = \Lambda$ and $\mathfrak{D}_{\sigma\xi}= \Lambda'$. Thus,
    \[\mathfrak{D}_{\zeta\xi} = \Lambda \triangle \Lambda' \subseteq \mathfrak{D}_{\sigma\eta}.\]
    The final assertion of the lemma is trivial since $\mathfrak{D}_{\sigma\zeta} = \mathfrak{D}_{\sigma\eta}=\Lambda$  exactly when $\zeta =  \sigma_{\Lambda^c}\eta_{\Lambda} = \eta$.
\end{proof}

The following examples highlight how $\mathscr{C}$ affects the structure coefficients of the algebra considering two extreme cases.

\begin{example}[Unique cluster] \label{exmpl:unique.cluster}
    Consider the case when $\mathbb{L}$ is the unique cluster  $\mathscr{C}_{\blacktriangle}:=\{\mathbb{L}\}$. Hence, if $\sigma \in \mathtt{E}^\eta$, then the offspring set $\Omega_{\sigma\eta}= \{\sigma, \eta\}$. Therefore, the structure coefficients of $\mathcal{A}(\mathscr{C}_{\blacktriangle},\mu,\Phi,\Omega)$ are given by
    \begin{align*}
        c_{\sigma\eta,\sigma} = \frac{h_{\mathfrak{D_{\sigma \eta}}}^\Phi(\sigma)}{h_{\mathfrak{D_{\sigma \eta}}}^\Phi(\sigma)+h_{\mathfrak{D_{\sigma \eta}}}^\Phi(\eta)} &= \frac{1}{1+h_{\mathfrak{D_{\sigma \eta}}}^\Phi(\eta)/h_{\mathfrak{D_{\sigma \eta}}}^\Phi(\sigma)}\\ &= \frac{1}{1+\exp\big(H_{\mathfrak{D_{\sigma \eta}}}^\Phi(\sigma)-H_{\mathfrak{D_{\sigma \eta}}}^\Phi(\eta)\big)},
    \end{align*}
    and , when $\sigma \neq \eta$,
    \begin{equation} \label{eq:coeff.replica.parents}
        c_{\sigma\eta,\eta} = \frac{1}{1+\exp\big(H_{\mathfrak{D_{\sigma \eta}}}^\Phi(\eta)-H_{\mathfrak{D_{\sigma \eta}}}^\Phi(\sigma)\big)} = 1-c_{\sigma\eta,\sigma}.
    \end{equation} \hfill $\blacktriangleleft$ 
\end{example}

The property that coefficients exhibit in \eqref{eq:coeff.replica.parents} is a consequence of the offspring of distinct $\sigma$ and $\eta$ being necessarily a replica of one of the parents. Then, in particular, \eqref{eq:coeff.replica.parents} holds whenever $\Omega_{\sigma\eta}=\{\sigma,\eta\}$. The next example provides us a case where the offspring copies the configuration of the parents at each site individually.

\begin{example}[Atomic clusters and two spin systems] \label{ex:atomic.cluster}
    The diametrically opposite case of the unique cluster given in \cref{exmpl:unique.cluster} is when $\mathscr{C}$ is the set of atomic clusters \[\mathscr{C}_{\odot}:= \big\{\{x\}: x \in \mathbb{L}\big\}.\] Let us also consider that the set of spins $S$ is such that $\vert S\vert =2$. Then $\Omega_{\sigma\eta}= \big\{\xi \sigma_{(\mathfrak{D}_{\sigma\eta})^c}: \xi \in S^{\mathfrak{D}_{\sigma\eta}}\big\}$ when $\sigma \in \mathtt{E}^\eta$. Note that $Z_{\mathfrak{D}_{\sigma\eta}}^\Phi(\sigma) = \sum_{\zeta \in \Omega_{\sigma\eta}}h_{\mathfrak{D}_{\sigma\eta}}^\Phi(\zeta)$.

    Therefore, the structure coefficients of $\mathcal{A}(\mathscr{C}_{\odot},\mu,\Phi,\Omega)$ are precisely the Gibbs measures
    \[c_{\sigma\eta,\zeta} = \mu_{\mathfrak{D}_{\sigma\eta}}(\zeta\mid\sigma).\]\hfill $\blacktriangleleft$ 
\end{example}

Now that we are more familiar with some properties of the algebras, let us characterize the non-associativity.

\begin{lemma}
    Let $\mathcal{A}=\mathcal{A}(\mathscr{C},\mu,\Phi,\Omega)$ be the $\mathscr{C}$-genetic Gibbs algebra generated by $\mu \in \mathscr{G}(\Phi)$ on $\Omega$. Then it is necessary and sufficient that $\operatorname{dim}(\mathcal{A})\geq 3$ for $\mathcal{A}$ be a non-associative algebra.
\end{lemma}
\begin{proof}
    Consider $\vert \Omega\vert \geq 3$ and let $\sigma, \eta, \zeta \in \Omega$ be three distinct configurations such that $\sigma,\eta \in \mathtt{E}^\zeta$. Suppose that
    \[
    e_\sigma\cdot(e_\eta\cdot e_\zeta) = (e_\sigma\cdot e_\eta)\cdot e_\zeta,\]
    then 
    \begin{equation} \label{eq:sum.non.associative}
    \sum_{\varpi \in \Omega_{\eta\zeta}}c_{\eta\zeta,\varpi}\sum_{\theta \in \Omega_{\varpi\sigma}}c_{\varpi\sigma,\theta~} e_\theta=\sum_{\omega \in \Omega_{\sigma\eta}}c_{\sigma\eta,\omega}\sum_{\xi \in \Omega_{\omega\zeta}}c_{\omega\zeta,\xi~} e_\xi.
    \end{equation}
    
    Let us choose $\mathfrak{D}_{\sigma\eta} = \{x\}$ and $\mathfrak{D}_{\eta\zeta} = \{y\}$, then $\mathfrak{D}_{\sigma\eta} =\{x,y\}$. Here, we may take $x\neq y$ when $\vert S\vert =2$, and $x=y$ otherwise. Observe that $\Omega_{\sigma \eta} = \{\sigma,\eta\}$, $\Omega_{\eta\zeta}= \{\eta,\zeta\}$, and $\sigma,\zeta \in \Omega_{\sigma\zeta}$. By expanding \eqref{eq:sum.non.associative}, the coefficients accompanying $e_\sigma$ yield 
    \begin{equation} \label{eq:constants.non.associative}
        (c_{\eta\zeta,\eta} - c_{\sigma\zeta,\sigma})c_{\sigma\eta,\sigma} + c_{\eta\zeta,\zeta}c_{\sigma\zeta,\sigma}=0.
    \end{equation}

    Fix $a=c_{\sigma\zeta,\sigma}$, $b=c_{\eta\zeta,\zeta}$, and $c= c_{\sigma\eta,\sigma}$. Note that since $\vert \Omega_{\eta\zeta}\vert =2$, $c_{\eta\zeta,\eta}= 1-b$. The properties from the Hamiltonian and Gibbs measures ensure us that $a,b,c \in (0,1)$. Note that it follows from $ab\neq0$ that $a+b \neq 1$. Hence, one has by \eqref{eq:constants.non.associative} that
    $c=\frac{ab}{a+b-1}$ with $c<1$. It implies that
    \[ab < a+b-1,\]
    then $a(1-b) + b-1 > 0$, and thence $a>1$.
    It contradicts the fact of $a \in (0,1)$. Therefore, the algebra is not associative when one can choose distinct $\sigma,\eta,\zeta \in \Omega$ with unitary $\mathfrak{D}_{\sigma\eta}$ and $\mathfrak{D}_{\eta\zeta}$, which covers all cases where $\operatorname{dim}(\mathcal{A})\geq 3$.

    The associativity when $\vert \Omega\vert =\operatorname{dim}(\mathcal{A}) < 3$ is an immediate consequence of the commutativity of $\mathcal{A}$.
    
\end{proof}

The theorem below provides a decomposition of the algebra $\mathcal{A}=\mathcal{A}(\mathscr{C},\mu,\Phi,\Omega)$ into a direct sum of ideals. Let us define \[\mathcal{F}_{\Omega}:= \{\mathtt{F}^\eta: \eta \in \Omega\}.\] Consider now $\widetilde{\sigma}: \mathcal{F}_{\Omega} \to \Omega$ to be a choice that fixes $\widetilde{\sigma}(\mathtt{F}) \in \Omega$ such that $\mathtt{F} = \mathtt{F}^{\widetilde{\sigma}(\mathtt{F})}$. Set $\widetilde{\Omega}:=\{\widetilde{\sigma}(\mathtt{F}): \mathtt{F} \in \mathcal{F}_{\Omega}\}$. We will see below that the axiom of choice is not necessary to define $\widetilde{\sigma}$ and that $\mathcal{F}_{\Omega}$ is a set of ideals of $\mathcal{A}$.

\begin{theorem}[Decomposition of $\mathcal{A}$ into a direct sum of Markov ideals] \label{thm:genetic.decomposition}
    Let $\mathcal{A}=\mathcal{A}(\mathscr{C},\mu,\Phi,\Omega)$ be a $\mathscr{C}$-genetic Gibbs algebra generated by $\mu \in \mathscr{G}(\Phi)$ on $\Omega$. Then, for all $\eta \in \Omega$, \[e_\eta \mathcal{A} = \mathtt{F}^\eta.\] The subalgebras $\mathtt{F} \in \mathcal{F}_{\Omega}$ are Markov ideals of $\mathcal{A}$ such that
    \[\mathcal{A} = \bigoplus_{\mathtt{F} \in \mathcal{F}_{\Omega}}\mathtt{F} = \bigoplus_{\eta \in \widetilde{\Omega}}\mathtt{F}^\eta.\]
    Moreover, each $\mathtt{F} \in \mathcal{F}_{\Omega}$ has countable basis and $\vert \widetilde{\Omega}\vert = \vert \mathcal{F}_{\Omega}\vert = 2^{\aleph_0}$ when $\mathbb{L}$ is countably infinite.
\end{theorem}
\begin{proof}
    Regard $\eta \in \Omega$ as fixed. Since for all $u \in \mathcal{A} \setminus \mathtt{F}^\eta$, $u$ is given by the formal sum \[u = \sum_{\sigma \in \Omega \setminus \mathtt{E}^{\eta}} a_\sigma e_\sigma,\]
    then $e_\eta \cdot u = 0$. It thus follows that $e_\eta \mathcal{A} \subseteq \mathtt{F}^\eta$. Observe that the idempotence of $e_\eta$ implies that $e_\eta \in e_\eta \mathcal{A}$. Suppose that, for an integer $k \geq 0$, if $\vert \mathfrak{D}_{\eta\sigma}\vert  \leq k$, then $e_\sigma \in e_\eta\mathcal{A}$. Let now $\zeta \in \mathtt{E}^\eta$ be such that $\vert \mathfrak{D}_{\eta\zeta}\vert  = k+1$. Then 
    \[e_\eta\cdot e_\zeta = \sum_{\xi \in \Omega_{\eta \zeta}} c_{\eta\zeta,\xi}e_\xi.\]

    However, one has by \cref{lm:discrep} and the assumption above that, for all $\xi \in \Omega_{\eta\zeta} \setminus \{\zeta\}$, $e_\xi \in e_\eta \mathcal{A}$. Since $c_{\eta\zeta, \zeta} \neq 0$, 
    \[e_\zeta = \frac{1}{c_{\eta\zeta, \zeta}}\left(\sum_{\xi \in \Omega_{\eta \zeta}\setminus\{\zeta\}} c_{\eta\zeta,\xi}e_\xi - e_\eta\cdot e_\zeta\right) \in e_\eta\mathcal{A}.\]

    Hence, one has by finite induction that $\mathfrak{B}_{\mathtt{E}^\eta} \subseteq e_\eta \mathcal{A}$ and, consequentially, $\mathtt{F}^\eta \subseteq e_\eta\mathcal{A}$. Therefore, $\mathtt{F}^\eta = e_\eta\mathcal{A}$. In fact, $e_\eta\mathtt{F}^\eta = \mathtt{F}^\eta$ by the same arguments. Let us verify that $\mathtt{F}^\eta$ is an ideal. 

    Suppose that $\sigma\in \Omega$ is such that $\sigma \not \in \mathtt{E}^\eta$ and let $v \in \mathtt{F}^\eta$. Since $e_\sigma \cdot e_\zeta =0$ for all $\zeta \in \mathtt{E}^\eta$, $e_\sigma \cdot v =0$. Thus, $e_\sigma \mathtt{F}^\eta = \langle0\rangle \subseteq \mathtt{F}^\eta$. On the other hand, $e_\zeta \mathtt{F}^\eta = \mathtt{F}^\eta$. It then follows that each $\mathtt{F} \in \mathcal{F}_{\Omega}$ is an ideal of $\mathcal{A}$.

    Observe that, for all $\zeta \in \mathtt{E}^\eta$, $\Omega_{\zeta\eta} \neq \emptyset$. Then $\sum_{\sigma\in \Omega_{\zeta\eta}}c_{\zeta\eta,\sigma}=1$, which implies that, for all $\eta \in \Omega$, $\mathtt{F}^\eta$ is Markov. It remains to verify that $\mathfrak{B}_{\mathtt{E}^\eta}$ is countable for all $\eta \in \Omega$. Note that every $\zeta \in \mathtt{E}^\eta$ is given by $\zeta = \eta_{\Lambda^c}\omega$ for a $\Lambda \in \mathcal{L}$ and $\omega \in S^\Lambda$. Hence,
    \[
        \left\vert\mathtt{E}^\eta\right\vert = \left\vert\bigcup_{\Lambda\in \mathcal{L}} S^\Lambda \right\vert.
    \]

    However, if $\{\Lambda_n\}_{n \in \mathbb{N}}$ is such that $\Lambda_n \uparrow \mathbb{L}$, then $\mathcal{L} = \bigcup_{n\in\mathbb{N}}2^{\Lambda_n}$, which implies that $\vert \mathcal{L}\vert =\aleph_0$. It then follows that $\mathfrak{B}_{\mathtt{E}^\eta}$ is countable. The direct sum is a straightforward from the definition of $\mathcal{A}$ as a free $\mathbb{K}$-module the Hamel basis $\mathfrak{B}_{\Omega}$. Note that $\vert \mathfrak{B}_{\Omega}\vert  \leq \vert \widetilde{\Omega}\times \mathfrak{B}_{\mathtt{E}^\eta}\vert $. Since $\vert \Omega\vert = 2^{\aleph_0}$ when $\vert \mathbb{L}\vert = \aleph_0$, one has that $\vert \widetilde{\Omega}\vert  = \vert \mathcal{F}_{\Omega}\vert = 2^{\aleph_0}$, and the proof is complete.
\end{proof}

The following results deal with isomorphism between genetic Gibbs algebras $\mathcal{A}$ and $\mathcal{A}'$ generated by Gibbs measures defined on the same $\Omega$. To differentiate them, we denote by $\mathfrak{B}_{\Omega}=\{e_\sigma\}_{\sigma\in \Omega}$ and $\mathfrak{B}_{\Omega}'=\{e_\sigma'\}_{\sigma\in \Omega}$ their respective standard basis. We also let $c_{\zeta\eta,\sigma}$ and $c_{\zeta\eta,\sigma}'$ denote the structure coefficients of $\mathcal{A}$ and $\mathcal{A}'$, respectively.

\begin{theorem}\label{thm:genetic.equiv.potentials}
    Let $\mathcal{A} = \mathcal{A}(\mathscr{C},\mu,\Phi,\Omega)$ and $\mathcal{A}'= \mathcal{A}(\mathscr{C},\mu',\Psi,\Omega)$ be two $\mathscr{C}$-genetic Gibbs algebras. If the interaction potentials $\Phi$ and $\Psi$ are equivalent, then algebras $\mathcal{A}$ and $\mathcal{A}'$ are isomorphic.
\end{theorem}
\begin{proof}
    Set $\varphi:\mathcal{A} \to \mathcal{A}'$ to be defined by linear extension of $\varphi(e_{\zeta\eta})=e_{\zeta\eta}'$ for all $\zeta, \eta\in \Omega$. Then $\varphi$ is an isomorphism of algebras when $\varphi$ is multiplicative. It suffices to show that $c_{\zeta\eta, \sigma}= c_{\zeta\eta, \sigma}'$ for all $\zeta, \eta,\sigma \in \Omega$. One has by definition that $c_{\zeta\eta, \sigma}= c_{\zeta\eta, \sigma}'=0$ for every $\sigma \not\in \Omega_{\zeta,\eta}$. 
    
    Consider now $\sigma \in \Omega_{\zeta\eta}$ and $\Phi \sim \Psi$, one has by \eqref{eq:DLR} and \eqref{eq:equiv.potentials} that
    \begin{equation*}
        c_{\zeta\eta,\sigma}=\frac{h_{\mathfrak{D_{\zeta \eta}}}^\Phi(\sigma)}{\sum\limits_{\xi \in \Omega_{\zeta\eta}}h_{\mathfrak{D_{\zeta \eta}}}^\Phi(\xi)}=\frac{h_{\mathfrak{D_{\zeta \eta}}}^{\Phi-\Psi}(\sigma)~h_{\mathfrak{D_{\zeta \eta}}}^\Psi(\sigma)}{h_{\mathfrak{D_{\zeta \eta}}}^{\Phi-\Psi}(\sigma)\sum\limits_{\xi \in \Omega_{\zeta\eta}}h_{\mathfrak{D_{\zeta \eta}}}^\Psi(\xi)}=c_{\zeta\eta,\sigma}',
    \end{equation*}
    and the proof is complete.
\end{proof}

The previous result establishes a connection of the interaction potentials and the algebras. A direct consequence is the following corollary. It states that the phase transition of an interaction potential does not affect the defined algebras.

\begin{corollary}[Stability under phase transition] \label{cor:stability.phase.transition}
    The phase transition of an interaction potential $\Phi$ does not affect the generated $\mathscr{C}$-genetic Gibbs algebra. In other words, if $\mu, \mu' \in \mathscr{G}(\Phi)$, then $\mathcal{A}(\mathscr{C},\mu,\Phi,\Omega)$ is isomorphic to $\mathcal{A}(\mathscr{C},\mu',\Phi,\Omega)$.    
\end{corollary}

As an example of Gibbs measures exhibiting phase transition, the Ising model on the hypercubic lattice $\mathbb{Z}^d$ with $d \geq 2$ at low temperatures (see Example \ref{ex:potts} and Georgii \cite{georgii2011}) provides a classical case where multiple Gibbs measures coexist. Nevertheless, Corollary \ref{cor:stability.phase.transition} ensures that the associated Gibbs algebras remain isomorphic.

We present below an example of transformation that preserves the algebra under isomorphism but not the equivalence of potentials. Let us first define $\mathcal{T}$ to be the set of all transformations $\tau:\Omega\to\Omega$ of the form
\begin{equation*}
    \tau\omega:=(\tau_x\omega_{\tau_\ast^{-1}x})_{x\in\mathbb{L}}
\end{equation*}
where $\tau_\ast: \mathbb{L} \to \mathbb{L}$ and $\tau_x: S \to S$ are bijections. Note that $\tau \in \mathcal{T}$ is invertible such that $\tau^{-1} \in \mathcal{T}$

The transformations $\tau \in \mathcal{T}$ permute the configurations of $\Omega$ preserving part of its original structure. They can be particularly interesting in the study of Gibbs measures determining, for instance, \textit{translations}, \textit{spin flip}, and \textit{spin rotation} (see \cite[Sec.~5.1]{georgii2011}). Before we proceed, we define $\tau(\Phi) := (\Phi_{\tau_\ast^{-1} A}\circ\tau^{-1})_{A \in \mathcal{L}}$ and observe that $\tau(\Phi)$ is also an interaction potential.

\begin{example}[Isomorphism under translation]
    Consider $\mathbb{L}=\mathbb{Z}$, $S=\{0,1\}$, and $\mathscr{C}=\mathscr{C}_{\blacktriangle} =\{\mathbb{Z}\}$. Set $\Phi=(\Phi_A)_{A \in \mathcal{L}}$ to be given by
    \[
        \Phi_A(\sigma) := \left\{\begin{array}{cc}
             \sigma(0),&  \text{if } A=\{0,1\}\\
             0,& \text{otherwise.} 
        \end{array}\right.
    \]
    for all $\sigma \in \Omega$. Let us define $\tau:\Omega\to\Omega$ to be the \textit{translation} (or \textit{shift}) such that $\tau \sigma(x) = \sigma(x-1)$ for all $ x \in \mathbb{Z}$. Then $\tau_\ast(x)= x+1$ and $\tau_x=\operatorname{id}$ for all $x \in \mathbb{Z}$. 
    
    Thus, for each $A \in \mathcal{L}$ and every $\sigma \in \Omega$,
    \[
        \tau(\Phi)_A(\sigma) := \left\{\begin{array}{cc}
             \sigma(1),&  \text{if } A=\{1,2\}\\
             0,& \text{otherwise.} 
        \end{array}\right.
    \]

    Observe that $H_{\{0\}}^\Phi(\sigma) = \sigma(0)$ while $H_{\{0\}}^{\tau(\Phi)} \equiv 0$, therefore $\Phi \not\sim \tau(\Phi)$. Now, let $\mathcal{A}= \mathcal{A}(\mathscr{C}_{\blacktriangle},\mu,\Phi,\Omega)$ and $\mathcal{A}'= \mathcal{A}(\mathscr{C}_{\blacktriangle},\mu',\tau(\Phi),\Omega)$. Set $\varphi: \mathcal{A} \to \mathcal{A}'$ to be given by linear extension of $\varphi(e_\sigma)=e_{\tau\sigma}'$ for all $\sigma \in \Omega$. Since $H_\Lambda^\Phi = H_{\tau_\ast\Lambda}^{\tau(\Phi)}\circ\tau$, one can easily verify that $\varphi$ is an isomorphism (see \cref{thm:genetic.T.equiv.potentials} for details). \hfill $\blacktriangleleft$ 
\end{example}

The example above gives us evidence that \cref{thm:genetic.equiv.potentials} can be improved by extending relations between potentials and isomorphisms via $\tau \in \mathcal{T}$. We call the potentials $\Phi$ and $\Psi$ $\mathcal{T}$\textit{-equivalent} when there exists $\tau \in \mathcal{T}$ such that $\Phi \sim \tau^{-1}(\Psi)$.

Consider $\tau \in \mathcal{T}$, let us define $\tau(\mathscr{C}) :=\{\tau_\ast \Delta : \Delta \in \mathscr{C}\}$ and observe that $\tau(\mathscr{C})$ is also a set of clusters. Let $\mathcal{A}$ and $\mathcal{A}'$ be two genetic algebras generated by Gibbs measures defined on the same $\Omega$. Recall the notation defined above. We say that the algebra $\mathcal{A}$ is $\tau$-\emph{isomorphic} to $\mathcal{A}'$ when the linear extension of $\varphi(e_\sigma)= e_{\tau\sigma}'$ is an isomorphism of algebras.

\begin{theorem}\label{thm:genetic.T.equiv.potentials}
    Let $\tau \in \mathcal{T}$ and let $\mathcal{A} = \mathcal{A}(\mathscr{C},\mu,\Phi,\Omega)$ and $\mathcal{A}'= \mathcal{A}\big(\tau(\mathscr{C}),\mu',\Psi,\Omega\big)$ be two genetic Gibbs algebras. If $\Phi\sim \tau^{-1}(\Psi)$, then the algebra $\mathcal{A}$ is $\tau$-isomorphic to $\mathcal{A}'$. Moreover, the converse holds when $\mathscr{C}=\mathscr{C}_{\odot}$ is the set of atomic clusters. 
\end{theorem}
\begin{proof}
    The fist part is similar to the proof of \cref{thm:genetic.equiv.potentials}. Since $\tau \in \mathcal{T}$ is given, it suffices to show that $c_{\zeta\eta,\sigma}=c_{\tau\zeta\tau\eta,\tau\sigma}'$ for all $\zeta,\eta,\sigma\in\Omega$. 
    
    Indeed, by the definition of $\tau \in \mathcal{T}$, one has $\tau_\ast\mathfrak{D}_{\zeta\eta} = \mathfrak{D}_{\tau\zeta\tau\eta}$. Therefore, if $\zeta \not \in \mathtt{E}^\eta$, then $\tau\zeta \not \in \mathtt{E}^{\tau\eta}$ and $c_{\zeta\eta,\sigma} = c_{\tau\zeta\eta,\sigma}'=0$ for all $\sigma \in \Omega$. Consider now $\zeta \in \mathtt{E}^\eta$. It is also straightforward that $\tau \Omega_{\zeta\eta}= \Omega_{\tau\zeta \tau\eta}$. Hence, for all $\sigma\in \Omega \setminus \Omega_{\zeta\eta}$, we also have that $c_{\zeta\eta,\sigma} = c_{\tau\zeta\eta,\sigma}'=0$. 
    
    Observe that, for all $\Lambda \in \mathcal{L}$,
    \begin{equation} \label{eq:Hamiltonian.tau.Psi}
        H_{\tau_\ast\Lambda}^{\Psi} \circ \tau = \sum_{A \in \mathcal{L},~A \cap \Lambda\neq \emptyset} \Psi_{\tau_\ast A}\circ \tau  = H_\Lambda^{\tau^{-1}(\Psi)}.
    \end{equation}
    
    Let $\sigma \in \Omega_{\zeta\eta}$ and $\Phi \sim \tau^{-1}(\Psi)$, then one has by \eqref{eq:DLR}, \eqref{eq:equiv.potentials}, and \eqref{eq:Hamiltonian.tau.Psi} that
    \begin{equation*}
        c_{\zeta\eta,\sigma}=\frac{h_{\mathfrak{D_{\zeta \eta}}}^\Phi(\sigma)}{\sum\limits_{\xi \in \Omega_{\zeta\eta}}h_{\mathfrak{D_{\zeta \eta}}}^\Phi(\xi)}=\frac{h_{\mathfrak{D_{\zeta \eta}}}^{\tau^{-1}(\Psi)}(\sigma)}{\sum\limits_{\xi \in \Omega_{\zeta\eta}}h_{\mathfrak{D_{\zeta \eta}}}^{\tau^{-1}(\Psi)}(\xi)}=\frac{h_{\mathfrak{D_{\tau\zeta \tau\eta}}}^{\Psi}(\tau\sigma)}{\sum\limits_{\xi' \in \Omega_{\tau\zeta\tau\eta}}h_{\mathfrak{D_{\tau\zeta \tau\eta}}}^{\Psi}(\xi')}=c_{\tau\zeta\tau\eta,\tau\sigma}'.
    \end{equation*}
    It follows that $\mathcal{A}$ is $\tau$-isomorphic to $\mathcal{A}'$.

    Suppose now that $\mathscr{C}=\mathscr{C}_\odot$ it the set of atomic clusters. Then, for all $\sigma \in \Omega_{\zeta\eta}$,
    \begin{equation*}
        c_{\zeta\eta,\sigma}=\frac{h_{\mathfrak{D_{\zeta \eta}}}^\Phi(\sigma)}{\sum\limits_{\xi \in \Omega_{\zeta\eta}}h_{\mathfrak{D_{\zeta \eta}}}^\Phi(\xi)}=\frac{h_{\mathfrak{D_{\tau\zeta \tau\eta}}}^{\Psi}(\tau\sigma)}{\sum\limits_{\xi' \in \Omega_{\tau\zeta\tau\eta}}h_{\mathfrak{D_{\tau\zeta \tau\eta}}}^{\Psi}(\xi')}=c_{\tau\zeta\tau\eta,\tau\sigma}'.
    \end{equation*}

    In particular, there exists $k_{\zeta\eta}>0$ depending on $\Omega_{\zeta\eta} $ such that
    \begin{equation} \label{eq:k.offspring}
        k_{\zeta\eta}:=\frac{h_{\mathfrak{D_{\zeta \eta}}}^\Phi(\sigma)}{h_{\mathfrak{D_{\tau\zeta \tau\eta}}}^{\Psi}(\tau\sigma)}=\frac{\sum\limits_{\xi \in \Omega_{\zeta\eta}}h_{\mathfrak{D_{\zeta \eta}}}^\Phi(\xi)}{\sum\limits_{\xi' \in \Omega_{\tau\zeta\tau\eta}}h_{\mathfrak{D_{\tau\zeta \tau\eta}}}^{\Psi}(\xi')}.
    \end{equation}

    Fix an arbitrary $\Lambda \in \mathcal{L}$ and let $\Theta_{\zeta,\Lambda}:=\{\eta \in \Omega: \mathfrak{D}_{\zeta,\eta}=\Lambda\}$. Since $\zeta,\eta \in \Omega_{\zeta\eta}$ for all $\eta \in \Theta_{\zeta,\Lambda}$, it follows from \eqref{eq:k.offspring} that
    \begin{equation} \label{eq:k.const.Lambda.offspring}
        h_\Lambda^\Phi(\zeta)=h_{\tau_\ast\Lambda}^\Psi(\tau\zeta) k_{\zeta\eta}=h_{\tau_\ast\Lambda}^\Psi(\tau\zeta) k_{\zeta\omega} \quad \text{for all }\eta,\omega \in \Theta_{\zeta,\Lambda}.
    \end{equation}
    Thus, $k_{\zeta,\eta}=k_{\zeta,\omega}$ for all $\eta,\omega \in \Theta_{\zeta,\Lambda}$. Let us fix $k^\zeta_\Lambda=k_{\zeta,\eta}$ for $\eta \in \Theta_{\zeta,\Lambda}$.

    Consider now $\sigma= \xi\zeta_{\Lambda^c}$ with $\xi \in S^\Lambda$. Then, since $\mathscr{C}=\mathscr{C}_\odot$, there exists $\eta \in \Theta_{\zeta,\Lambda}$ such that $\sigma \in \Omega_{\zeta\eta}$. By \eqref{eq:k.const.Lambda.offspring}, one has that
    $h_\Lambda^\Phi(\sigma)= h_{\tau_\ast\Lambda}^\Psi(\tau\sigma)k_\Lambda^{\zeta}$. It then follows from \eqref{eq:Hamiltonian.tau.Psi} that, for all $\xi \in S^\Lambda$,
    \[
        k_\Lambda^\zeta = \frac{h_\Lambda^\Phi(\xi\zeta_{\Lambda^c})}{h_{\tau_\ast\Lambda}^\Psi(\tau~\xi\zeta_{\Lambda^c})} = \frac{h_\Lambda^\Phi(\xi\zeta_{\Lambda^c})}{h_\Lambda^{\tau^{-1}(\Psi)}(\xi\zeta_{\Lambda^c})},
    \]
    which in turn proves that $H_\Lambda^{\Phi}-H_\Lambda^{\tau^{-1}(\Psi)}$ is $\mathscr{T}_{\Lambda}$-measurable for all $\Lambda \in \mathcal{L}$. In other words, $\Phi \sim \tau^{-1}(\Psi)$.
\end{proof}

The property presented by $\mathscr{C}_\odot$ in \cref{thm:genetic.T.equiv.potentials} shows that we may recover  information regarding the interaction potential from the algebra when $\mathscr{C}=\mathscr{C}_\odot$. This would be expected when $\vert S\vert =2$ due to the property presented in \cref{ex:atomic.cluster} but \eqref{eq:k.const.Lambda.offspring} allows us to compare all configurations varying in a given region. 

We will show a result for some functionals of the $\mathscr{C}$-genetic Gibbs algebra. Denote by $\hom(\mathcal{A},\mathbb{K})$ the set of multiplicative linear functionals of the $\mathbb{K}$-algebra $\mathcal{A}=\mathcal{A}(\mathscr{C},\mu,\Phi,\Omega)$ to $\mathbb{K}$, \textit{i.e.}, the set of linear functionals $\varphi:\mathcal{A}\to\mathbb{K}$ such that $\varphi(u\cdot v)=\varphi(u)\cdot\varphi(v)$. Let $\pi_\mathtt{F}$ stand for the multiplicative linear functional such that, for any $\sigma\in\Omega$, $\pi_\mathtt{F}(e_\sigma)=1$ when $e_\sigma\in \mathtt{F} \subset \mathcal{A}$ and $\pi_\mathtt{F}(e_\sigma)=0$ otherwise.

\begin{theorem}[Multiplicative functionals of $\mathcal{A}$] \label{thm:hom.A}

    Let $\mathcal{A} = \mathcal{A}(\mathscr{C},\mu,\Phi,\Omega)$ be a $\mathscr{C}$-genetic Gibbs algebra generated by $\mu\in\mathscr{G}(\Phi)$ on $\Omega$. Then
    \begin{equation*}
    \hom(\mathcal{A},\mathbb{K})=\Big\{\pi_\mathtt{F}:\mathtt{F}\in\mathcal{F}_{\Omega}\cup\{\langle0\rangle\}\Big\}.
    \end{equation*}
    \end{theorem}
    
\begin{proof}
    It is straightforward to see that $\pi_{\langle0\rangle}\equiv0 \in hom(\mathcal{A},\mathbb{K})$. Let $\varphi\in\hom(\mathcal{A},\mathbb{K})\setminus\{\pi_{\langle0\rangle}\}$. Since every $e_\eta \in \mathfrak{B}_{\Omega}$ is idempotent, one has that, for every $\eta \in \Omega$,
    \begin{equation*}
    \varphi(e_\eta)=\varphi(e_\eta^2)=\varphi(e_\eta)^2.
    \end{equation*}
    It follows that $\varphi(e_\eta)\in\{0,1\}$. In fact, there exists $\eta \in \Omega$ such that $\varphi(e_\eta)=1$. If $\eta\notin \mathtt{E}^\sigma$, $e_\eta e_\sigma=0$ and, since $\mathbb{K}$ is an integral domain, $\varphi\vert _{\mathtt{F}^\sigma}=0$. 
        
    Observe now that if $\zeta\in \mathtt{E}^\eta$,
    \[
        \varphi(e_\zeta \cdot e_\eta) = \sum_{\xi \in \Omega_{\zeta\eta}}c_{\zeta\eta,\xi} \varphi(e_\xi)
    \]
    with $c_{\zeta\eta,\xi} \in (0,1)$ and $\varphi(e_\xi) \in \{0,1\}$ for all $\xi \in \Omega_{\zeta\eta}$. We also have that $\varphi(e_\zeta)\varphi(e_\eta) \in \{0,1\}$. Note that $\eta \in \Omega_{\zeta\eta}$, then $\varphi(e_\zeta)\varphi(e_\eta) \geq c_{\zeta\eta,\eta}$, which implies that
    \[
        \varphi(e_\zeta \cdot e_\eta) = \sum_{\xi \in \Omega_{\zeta\eta}}c_{\zeta\eta,\xi} =1.
    \]
    
    Therefore, $\varphi(e_\zeta)=1$ for all $\zeta \in \mathtt{E}^\eta$, $\varphi=\pi_{\mathtt{F}^\eta}$ and $\pi_{\mathtt{F}^\eta}$ is well-defined by linear extension.
\end{proof}

The theorem above shows us that all weight homomorphisms of $\mathcal{A}=\mathcal{A}(\mathscr{C},\mu,\Phi,\Omega)$ are determined by $\pi_\mathtt{F}$ with $\mathtt{F}\in \mathcal{F}_{\Omega}$. Hence, for all $\mathtt{F} \in \mathcal{F}_{\Omega}$,  $(\mathcal{A},\pi_\mathtt{F})$ is a weighted algebra. In particular, if $\mathbb{L}$ is finite, then $(\mathcal{A}, \pi_\mathcal{A})$ is the unique possible weighted algebra for $ \mathcal{A}$.

Let us introduce a new notation to extend our results to product probability spaces. Consider $\Phi^i$ an interaction potential on $\Omega_i = S^{\mathbb{L}_i}$ for $i \in I$. Set $\Omega = \prod_{i \in I} \Omega_i$ and $\mathcal{L}$ the se of finite subsets of $\mathbb{L} = \bigsqcup_{i \in I} \mathbb{L}_i$. Define $\bigoplus_{i\in I} \Phi^i = (\Phi_A)_{A \in \mathcal{L}}$ to be given by 
\begin{equation} \label{eq:def.sum potentials}
    \Phi_A(\sigma) := \sum_{i \in I}\Phi^i_{A_i}(\sigma_{\mathbb{L}_i})
\end{equation}
for each $\sigma \in \Omega$ and $A = \bigsqcup_{i \in I}A_i$ with $A_i \subseteq \mathbb{L}_i$. We can now state a theorem that establish a relation between finite products of Gibbs probability spaces and tensor algebras. From the probabilistic point of view, the product measure indicates independence which, in terms of Gibbs measures, is a lack interaction between certain regions. We will see below that this property is translated as the tensor product by the genetic algebra.

\begin{theorem}[Genetic algebras generated by products of Gibbs measures] \label{thm:genetic.product.v.tensor.alg}
    Let $\{\mathbb{L}_i\}_{i=1}^n$ be a sequence of countable sets such that, for each $i\in \{1, \dots, n\}$, $\mathscr{C}_i$ is a partition of $\mathbb{L}_i$ associated with a Gibbs measure $\mu_i \in \mathscr{G}(\Phi^i)$ on $\Omega_i = S^{\mathbb{L}_i}$ with $S$ a fixed finite set of spins.
    
    Consider the genetic Gibbs algebras $\mathcal{A}_i:= \mathcal{A}(\mathscr{C}_i,\mu_i,\Phi^i,\Omega_i)$ for  all $i \in \{1, \dots, n\}$,  and \[\mathcal{A}=\mathcal{A}\left(\bigsqcup_{i=1}^n \mathscr{C}_i,\bigotimes_{i=1}^n\mu_i,\bigoplus_{i=1}^n\Phi^i,\prod_{i=1}^n\Omega_i\right).\] Then $\mathcal{A}$ is isomorphic to the tensor algebra $\bigotimes_{i =1}^n \mathcal{A}_i$ equipped with the ordinary product.
\end{theorem}
\begin{proof}
    First, we verify that $\mathcal{A}$ is well-defined. One can easily see that the disjoint union of cluster sets $\mathscr{C}:= \bigsqcup_{i=1}^n \mathscr{C}_i$ is a partition of $\mathbb{L}= \bigsqcup_{i=1}^n \mathbb{L}_i$. Note that $\Omega :=\prod_{i=1}^n\Omega_i = S^{\mathbb{L}}$. We will now prove that $\Phi:=\bigoplus_{i=1}^n\Phi^i$ is an admissible potential and the product measure $\mu:= \bigotimes_{i=1}^n\mu_i$ is such that $\mu \in \mathscr{G}(\Phi)$.

    Let $\sigma \in \Omega$ and $\Lambda \in \mathcal{L}$ be such that $\Lambda = \bigsqcup_{i=1}^n \Lambda_i$ with $\Lambda_i \subseteq \mathbb{L}_i$. Then, using the notation and definition from \eqref{eq:def.sum potentials}, 
    \[H_\Lambda^\Phi(\sigma) = \sum_{A \in \mathcal{L}, ~A \cap \Lambda \neq \emptyset}\left(\sum_{i=1}^n \Phi_A(\sigma_{\mathbb{L}_i})\right) = \sum_{i=1}^n H_{\Lambda_i}^{\Phi^i}(\sigma_{\mathbb{L}_i}).\]

    It follows that $\Phi$ is admissible and $h_\Lambda^\Phi(\sigma) = \prod_{i=1}^n h_{\Lambda_i}^{\Phi^i}(\sigma_{\mathbb{L}_i})$ for all $\sigma \in \Omega$ and $ \Lambda \in \mathcal{L}$. Hence,
    \[Z_\Lambda^\Phi = \sum_{\sigma \in S^\Lambda} h_\Lambda^\Phi(\sigma) = \sum_{\sigma_1 \in S^{\Lambda_1}}\cdots\sum_{\sigma_n \in S^{\Lambda_n}}\prod_{i=1}^{n}h_{\Lambda_i}^{\Phi^i}(\sigma_i) = \prod_{i=1}^n Z_{\Lambda_i}^{\Phi^i}.\]
    Observe that the $Z_\Lambda^\Phi $ is well-defined even when a $\Lambda_i = \emptyset$, since $H_{\emptyset}^{\Phi^i}\equiv0$, $h_{\emptyset}^{\Phi^i}\equiv1$ and $S^{\emptyset}=\{\emptyset\}$. One has that, for all  $\sigma, \xi \in \Omega$,
    \begin{align}
    \mu_\Lambda(\sigma\mid \xi) = \frac{h_\Lambda^\Phi(\sigma) }{Z_\Lambda^\Phi}\mathds{1}_{\{\sigma\}}(\sigma_\Lambda\xi_{\Lambda^c})&= \prod_{i=1}^n\frac{h_{\Lambda_i}^{\Phi^i}(\sigma_{\mathbb{L}_i}) }{Z_{\Lambda_i}^{\Phi^i}}\mathds{1}_{\{\sigma_{\mathbb{L}_i}\}}(\sigma_{\Lambda_i}\xi_{\mathbb{L}_i \setminus \Lambda_i}) \nonumber\\
    &= \prod_{i=1}^n \mu_i(\sigma_{\mathbb{L}_i}\mid\xi_{\mathbb{L}_i \setminus \Lambda_i}) = \left(\bigotimes_{i=1}^n \mu_i\right)(\sigma\mid\xi_{\Lambda^c}). \label{eq:equality.product.measure}
    \end{align}
    Therefore, $\mu$ is a Gibbs measure such that $\mu \in \mathscr{G}(\Phi)$. We now proceed to prove the isomorphism between $\mathcal{A}$ and the tensor algebra $\bigotimes_{i=1}^n \mathcal{A}_i$.

    Let $\sigma, \eta, \zeta \in \Omega$ and write $\sigma= \sigma_1 \cdots \sigma_n$, $\eta= \eta_1 \cdots \eta_n$, and  $\zeta= \zeta_1 \cdots \zeta_n$ with $\sigma_i, \eta_i, \zeta_i \in \Omega_i$ for all $i \in \{1, \dots, n\}$. Since $\mathscr{C}$ is the disjoint union of the sets of clusters $\mathscr{C}_i$'s, every $\zeta \in \Omega_{\sigma\eta}$ is determined by $\zeta_i \in \Omega_{\sigma_i\eta_i}$ for each $i \in \{1, \dots, n\}$, then one can easily see that 
    \[\Omega_{\sigma\eta} = \prod_{i=1}^n \Omega_{\sigma_i\eta_i}.\]
    The equality above holds even when $\sigma \not \in \mathtt{E}^\eta$, since there exists $\Omega_{\sigma_i\eta_i} = \emptyset$ and, necessarily, $\Omega_{\sigma\eta}= \emptyset$.

    Consider $\varphi: \mathcal{A} \to \bigotimes_{i=1}^n\mathcal{A}_i$ be defined by linear extension of $\varphi(e_\sigma) = \varphi(e_{\sigma_1 \cdots \sigma_n}) = e_{\sigma_1}\otimes \cdots \otimes e_{\sigma_n}$. The ordinary product of tensor algebra $\bigotimes_{i=1}^n\mathcal{A}_i$  and \eqref{eq:equality.product.measure} imply that
    \begin{align*}
        \varphi(e_\sigma e_\eta) =\varphi(e_\sigma)\varphi(e_\eta) &= (e_{\sigma_1}\otimes \cdots \otimes e_{\sigma_n}) \cdot (e_{\eta_1}\otimes \cdots \otimes e_{\eta_n})\\ &= \left(\sum_{\zeta_1 \in \Omega_{\sigma_1\eta_1}}c_{\sigma_1\eta_1,\zeta_1} e_{\zeta_1}\right) \otimes \cdots \otimes \left(\sum_{\zeta_n \in \Omega_{\sigma_n\eta_n}}c_{\sigma_n\eta_n,\zeta_n} e_{\zeta_n}\right)\\
        &= \sum_{\zeta_1 \in \Omega_{\sigma_1\eta_1}}\cdots\sum_{\zeta_n \in \Omega_{\sigma_n\eta_n}}\prod_{i=1}^n c_{\sigma_i\eta_i,\zeta_1} ~ e_{\zeta_i}\otimes \cdots \otimes e_{\zeta_n}\\
        &= \sum_{\zeta \in \Omega_{\sigma\eta}} c_{\sigma\eta,\zeta} ~ \varphi(e_\zeta),
    \end{align*}
    which proves the theorem.
\end{proof}

The finite subalgebras of $\mathcal{A}(\mathscr{C},\Phi,\Omega)$ may exhibit a similar structure to the finite Gibbs algebras. Define $\mathscr{C}\vert_{L} := \{\Delta \cap L : \Delta\in \mathscr{C}\}\setminus\{\emptyset\}$ to be a partition of  $L \subseteq \mathbb{L}$ given by $\mathscr{C}$. We show below a special relation when they are a direct sum of subalgebras of Markov ideals $\mathtt{F}\in \mathcal{F}_{\Omega}$.

\begin{theorem} \label{thm:finite.subalgebras}
    Let $\mathcal{A}_f$ be a finite-dimensional subalgebra of $\mathcal{A}(\mathscr{C},\Phi,\Omega)$ such that the interaction potential $\Phi$ has finite range and
    \[\mathcal{A}_f = \bigoplus_{i=1}^n~ \mathtt{A}_i, \quad \text{with } \mathtt{A}_i \text{ subalgebras of distinct } \mathtt{F}_i \in \mathcal{F}_{\Omega}.\] Then there exists $\Lambda \in \mathcal{L}$ such that each $\mathtt{A}_i$ is isomorphic to a subalgebra of the finite-dimensional $\mathscr{C}\vert_\Lambda$-genetic Gibbs algebra $\mathlcal{A}_f(\mathscr{C}\vert_\Lambda,\mu_\Lambda(\cdot\mid\xi),S^\Lambda)$ with any fixed $\xi \in \Omega$. Moreover, the same holds for every $\mathlcal{A}_f(\mathscr{C}\vert_{\Lambda'},\mu_{\Lambda'}(\cdot\mid\xi),S^{\Lambda'})$ with $\Lambda \subseteq \Lambda' \in \mathcal{L}$.
\end{theorem}
\begin{proof}
    First, let $\mathtt{E}_i=\{\eta \in \mathtt{E}^{\widetilde{\sigma}(\mathtt{F}_i)}: e_\eta \in \mathtt{A}_i\}$, then $\mathtt{A}_i = \langle\mathfrak{B}_{\mathtt{E}_i}\rangle$ for every $i \in \{1, \dots, n\}$. Observe that, for all $\sigma,\eta \in \mathtt{E}_i$, $\mathfrak{D}_{\sigma\eta} \in \mathcal{L}$.
    Set 
    \[\Lambda := \operatorname{cl}_{\Phi}\left( \bigcup_{i=1}^n\left(\bigcup_{\sigma,\eta \in \mathtt{E}_i} \mathfrak{D}_{\sigma\eta}\right)\right).\]
    Hence, since $\Phi$ has finite range, $\Lambda \in \mathcal{L}$. Let us now define the algebra $\mathlcal{A}_{(\Lambda)}:=\mathlcal{A}_f(\mathscr{C}\vert_\Lambda,\mu_\Lambda(\cdot\mid\xi),S^\Lambda)$ as in Sec.\ref{sec:fin.dim.Gibbs.gen}  with a fixed $\xi \in \Omega$.

    Consider $\Lambda_i := \operatorname{cl}_\Phi\left(\bigcup_{\sigma,\eta \in \mathtt{E}_i} \mathfrak{D}_{\sigma\eta}\right)$ and choose $\xi^i \in \Omega$ arbitrarily fixed for each $i \in \{1, \dots, n\}$. Define $f_i: \mathtt{E}_i \to S^\Lambda$ such that $f_i(\sigma) = \sigma_{\Lambda_i}\xi^i_{\Lambda\setminus\Lambda_i}$. Set $\varphi_i: \mathtt{A}_i \to \mathlcal{A}_{(\Lambda)}$ to be given by linear extension of
    \[
        \varphi_i(e_\sigma) = e_{f_i(\sigma)}.
    \]

    Let us now prove that $\operatorname{Im}(\varphi_i)$ is a subalgebra of $\mathlcal{A}_{(\Lambda)}$ and $\varphi_i$ determines an isomorphism between $\mathtt{A}_i$ and $\operatorname{Im}(\varphi_i)$. It follows from the definitions of $f_i$  and $\mathscr{C}\vert_\Lambda$ that $\mathfrak{D}_{f_i(\sigma)f_i(\eta)}= \mathfrak{D}_{\sigma\eta}$ and $\Omega_{f_i(\sigma)f_i(\eta)} = \big\{f_i(\zeta):\zeta \in \Omega_{\sigma\eta}\big\}$. Note that, for all $\sigma,\eta \in \mathtt{E}_i$ and $\zeta \in \Omega_{\sigma\eta}$, since $\Phi_A$ is $\mathscr{F}_A$-measurable,
    \[H_{\mathfrak{D}_{\sigma\eta}}^\Phi(\zeta) = \sum_{A \in \mathcal{L}, ~ A \cap \mathfrak{D}_{\sigma\eta} \neq \emptyset}\Phi_A(\zeta)=\sum_{A \in \mathcal{L}, ~ A \cap \mathfrak{D}_{\sigma\eta} \neq \emptyset}\Phi_A\big(f_i(\zeta)\big) = H_{\mathfrak{D}_{f_i(\sigma)f_i(\eta)}}^\Phi\big(f_i(\zeta)\big).\]
    
    Therefore, by \cref{lm:finite.alg.equiv}, $c_{\sigma\eta,\zeta} = a_{f_i(\sigma)f_i(\eta),f_i(\zeta)}$ for all $\sigma,\eta \in \mathtt{E}_i$ and $\zeta \in \Omega_{\sigma\eta}$. Hence, by properties inherited from $\mathtt{A}_i$, we conclude that $\operatorname{Im}(\varphi_i)$ is a subalgebra of $\mathlcal{A}_{(\Lambda)}$ and $\varphi_i$ is an isomorphism of algebras.

    The result for $\mathlcal{A}_f(\mathscr{C}\vert_{\Lambda'},\mu_{\Lambda'}(\cdot\mid\xi),S^{\Lambda'})$ immediately obtained by replacing $\Lambda$ with $\Lambda' \in \mathcal{L}$ with $\Lambda \subseteq \Lambda'$ and following the same steps as above.
\end{proof}
\begin{figure}[!htb]
    \centering
    \includegraphics[width=200pt]{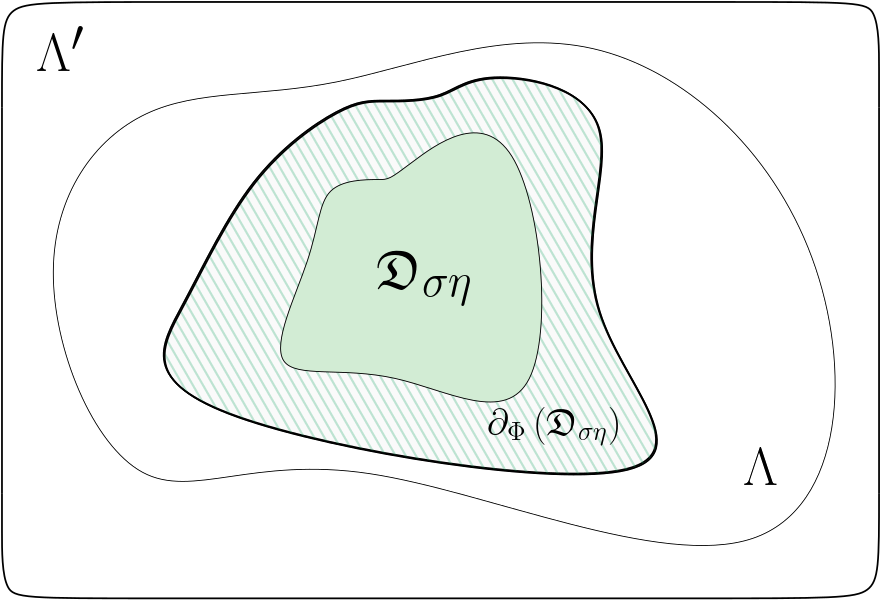}
    \caption{Regions determined by \cref{thm:finite.subalgebras} that eliminate the boundary effect on finite subalgebras.}
    \label{fig:finite.range.subalgebras}
\end{figure}

The theorem above may give us the impression that the the ideals $\mathtt{F} \in \mathcal{F}_{\Omega}$ could be isomorphic. However, it simply exhibits a finite algebraic compatibility between them. On the other hand, we will construct the evolution algebras in the next section so that we can establish the isomorphism of fertile classes when the potential has finite range (see \cref{thm:isomorphic.decomposition.E.M}).


\section{Infinite-dimensional Gibbs evolution algebras}

We will employ an approach similar to that used in genetic algebras to define evolution algebras generated by Gibbs measures. The main difference is that the evolution dynamics is determined by ordered pairs of configurations. The self-reproduction of a $(\sigma,\eta) \in \Omega^2$ will be compared with the coupling of its components $\sigma$ and $\eta$. The offspring produced by $(\sigma,\eta)$ is now given by the set $\Omega_{\sigma\eta}^2$.

Let $e_{\sigma\eta}$ stand for $e_{(\sigma,\eta)} \in \mathfrak{B}_{\Omega^2}$. The $\mathscr{C}$-evolution Gibbs algebra generated by $\mu \in \mathscr{G}(\Phi)$ on $\Omega$ is the free $\mathbb{K}-$module $\mathcal{E}(\mathscr{C},\mu,\Phi,\Omega)=\left\langle \mathfrak{B}_{\Omega^2} \right\rangle$ with product given by bilinear extension of
\[
    e_{\sigma\eta} \cdot e_{\sigma'\eta'} = \left\lbrace \begin{array}{cl}
    \sum\limits_{(\zeta,\xi) \in \Omega_{\sigma\eta}^2} \mathsf{c}_{\sigma \eta, \zeta\xi} ~e_{\zeta\xi}, & \text{ if } \sigma = \sigma' \text{ and } \eta=\eta';\\
     0, & \text{otherwise}. 
    \end{array}\right.
\]
where
\[\mathsf{c}_{\sigma \eta, \zeta\xi} =  \frac{\mu(\zeta\mid\sigma_{(\mathfrak{D}_{\sigma\eta})^c})\mu(\xi\mid\sigma_{(\mathfrak{D}_{\sigma\eta})^c})}{\mu^{\otimes 2}(\Omega_{\sigma\eta}^2\mid\sigma_{(\mathfrak{D}_{\sigma\eta})^c})} = c_{\sigma\eta,\zeta}c_{\sigma\eta,\xi}.\]

The lemma below proves that or definition is equivalent to the finite-dimensional evolution Gibbs algebras when $\mathbb{L}$ is finite. Hence, the results obtained in this section only require $\mathbb{L}$ to be countable.

\begin{lemma} \label{lm:finite.evo.alg.equiv}
    Let $\mathbb{L}$ and $S$ be finite. Consider $\Phi$ an admissible interaction potential and $\mu \in \mathscr{G}(\Phi)$ a Gibbs measure. Then, for every fixed set of clusters $\mathscr{C}$,
    \[\mathcal{E}(\mathscr{C},\mu,\Phi,\Omega) \cong \mathlcal{E}_f(\mathscr{C},\mu,\Omega).\]
\end{lemma}
\begin{proof}
    We follow the same steps from the proof of \cref{lm:finite.alg.equiv} replacing $\mathfrak{B}_{\Omega}$ by $\mathfrak{B}_{\Omega^2}$. Observe that from \eqref{eq:struct.coeff.finite}, for all $\zeta,\xi \in \Omega_{\sigma\eta}$,
    \[\frac{\mu(\zeta)\mu(\xi)}{\mu^{\otimes2}(\Omega_{\sigma\eta}^2)}= \frac{\mu(\zeta\mid\sigma_{(\mathfrak{D}_{\sigma\eta})^c})\mu(\xi\mid\sigma_{(\mathfrak{D}_{\sigma\eta})^c})}{\mu^{\otimes 2}(\Omega_{\sigma\eta}\mid\sigma_{(\mathfrak{D}_{\sigma\eta})^c})} = c_{\sigma\eta,\zeta}c_{\sigma\eta,\xi} = \mathsf{c}_{\sigma \eta, \zeta\xi}\]
    and the lemma holds.
\end{proof}

The genetic variation described by $\mathcal{E} =\mathcal{E}(\mathscr{C},\mu,\Phi,\Omega)$ in successive generations depends on $\mathscr{C}$ to define the offspring of the pairs $\sigma, \eta \in \Omega$. On the other hand, $\mu \in \mathscr{G}(\Phi)$ determines the proportion of the surviving offspring. This dynamics is described by $e_{\sigma\eta}^2$ in evolution algebras. If $\sigma \not \in \mathtt{E}^\eta$, $e_{\sigma\eta}^2=0$ indicates that infertility and $(\sigma,\eta)$ does not produce surviving offspring. Observe that the idempotent elements of $\mathcal{E}$ are exactly $\{e_{\sigma\sigma}\}_{\sigma \in \Omega}$ and their finite sums. Therefore $e_{\sigma\sigma}^2=e_{\sigma\sigma}$ indicates the persistence  and stability of $(\sigma,\sigma)$ across generations.

Consider now $\sigma,\eta \in \Omega$ to be distinct such that $\sigma \in \mathtt{E}^\eta$. Then $\vert\Omega_{\sigma\eta}\vert \geq 2$ and since $\sigma, \eta \in \Omega_{\sigma \eta}$, the genealogical tree of $e_{\sigma \eta}$ exhibits a self-similar structure of the gene flow. Furthermore, $e_{\sigma\eta}^2= e_{\eta\sigma}^2$ the same behavior is repeated for each pair $\zeta,\xi \in \Omega_{\sigma\eta}$. Therefore, the tree is indeed fractal. We represent  genealogical tree of $e_{\sigma\eta}$ in \cref{fig:selfsimilar} when $\Omega_{\sigma\eta}= \{\sigma,\eta\}$.
\begin{figure}[!htb]
    \centering
     \includegraphics[width=320pt]{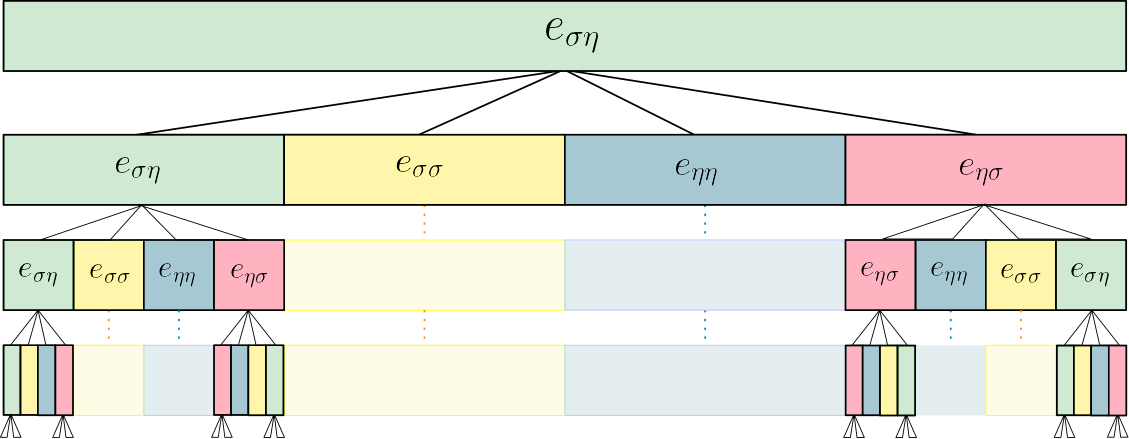}
     \caption{Self-similarity of the genealogical tree of $e_{\sigma\eta}$ when $\vert\Omega_{\sigma\eta}\vert=2$ and the stability of the idempotents.}
     \label{fig:selfsimilar}
\end{figure}

It can be easily seen that the structure coefficients of the $\mathscr{C}$-genetic and evolution Gibbs algebras are correlated. The theorem below establishes a characterization of the evolution algebra via isomorphism derived from the tensor algebra $\mathcal{A} \otimes\mathcal{A}$. 

\begin{theorem} \label{thm:evolution.iso.tensor}
    Let $\mathcal{A}= \mathcal{A}(\mathscr{C},\mu,\Phi,\Omega)$ and $\mathcal{E}= \mathcal{E}(\mathscr{C},\mu,\Phi,\Omega)$. Then $\mathcal{E}$ is isomorphic to the tensor module $\mathcal{A} \otimes \mathcal{A}$ with product $\ast$ given by bilinear extension of
    \[(e_\sigma \otimes e_\eta) \ast (e_{\sigma'} \otimes e_{\eta'}) := \delta_{\sigma \sigma'} \delta_{\eta\eta'}  (e_\sigma \otimes e_\eta) \cdot (e_{\eta} \otimes e_{\sigma}),\]
    where $\delta$ is the Kronecker delta and $\cdot$ denotes the ordinary product of the tensor algebra $\mathcal{A} \otimes \mathcal{A}$. 
\end{theorem}
\begin{proof}
    Set $\varphi:\mathcal{E} \to \mathcal{A} \otimes \mathcal{A}$ to be the linear bijection such that, for all $e_{\sigma\eta} \in \mathfrak{B}_{\Omega^2}$, one has that $\varphi(e_{\sigma\eta})=e_\sigma \otimes e_\eta$. Observe that
    \begin{align*}
    \varphi(e_{\sigma \eta}) \ast \varphi(e_{\sigma' \eta'}) &= (e_{\sigma} \otimes e_{\eta}) \ast (e_{\sigma'} \otimes e_{\eta'})\\
    &= \delta_{\sigma\sigma'}\delta_{\eta\eta'} (e_{\sigma}\cdot e_{\eta}) \otimes (e_{\sigma}\cdot e_{\eta})\\
    &= \delta_{\sigma\sigma'}\delta_{\eta\eta'} \left( \sum_{\zeta \in \Omega_{\sigma\eta}} c_{\sigma\eta,\zeta}e_\zeta\right) \otimes \left( \sum_{\xi \in \Omega_{\sigma\eta}} c_{\sigma\eta,\xi}e_\xi\right)\\
    &= \delta_{\sigma\sigma'} \delta_{\eta\eta'} ~ \varphi\left( \sum_{(\zeta,\xi) \in \Omega_{\sigma\eta}^2} c_{\sigma\eta,\zeta} c_{\sigma\eta,\xi}e_{\zeta\xi} \right)=  ~\varphi\left(e_{\sigma\eta} \cdot e_{\sigma'\eta'} \right).
  \end{align*}
  
  Thus the proof is completed by linear and bilinear extension of $\varphi$ and $\ast$, respectively.
\end{proof}

Let us define 
\[\mathtt{F}_{\sigma{\eta}}=\langle \mathfrak{B}_{\mathtt{E}^\sigma \times\mathtt{E}^\eta} \rangle.\]

We can now apply previous results to obtain the following decomposition of the algebra.

\begin{corollary} \label{cor:evo.alg.decomposition}
    Let $\mathcal{E}=\mathcal{E}(\mathscr{C}, \mu,\Phi, \Omega)$ be a $\mathscr{C}$-evolution Gibbs algebra generated by $\mu \in \mathscr{G}(\Phi)$ on $\Omega$. Then
    \[\mathcal{E} = \bigoplus_{(\sigma,\eta) \in \widetilde{\Omega}^2}\mathtt{F}_{\sigma\eta}\]
    where $\mathtt{F}_{\sigma{\eta}}$ is an ideal isomorphic to $\mathtt{F}^\sigma\otimes\mathtt{F^\eta}$ with the isomorphism given in \cref{thm:evolution.iso.tensor}.
\end{corollary}
\begin{proof}
    The direct sum and the isomorphism between $\mathtt{F}_{\sigma\eta}$ and  $\mathtt{F}^\sigma\otimes\mathtt{F^\eta}$ are an immediate consequences of \cref{thm:genetic.decomposition,thm:evolution.iso.tensor}. It remains to verify that, for all $\sigma,\eta \in \Omega$, $\mathtt{F}_{\sigma\eta}$ is an ideal of $\mathcal{E}$.

    Consider $u := \sum_{(\sigma',\eta') \in \Xi'}a_{\sigma'\eta'} e_{\sigma'\eta'} \in \mathtt{F}_{\sigma\eta}$ and let $v:= \sum_{(\sigma'',\eta'') \in \Xi''}a_{\sigma''\eta''}e_{\sigma''\eta''} \in \mathcal{E}$ be arbitrary. Then,
    \[u\cdot v= \sum_{(\sigma',\eta') \in \Xi' \cap \Xi''} a_{\sigma'\eta'}^2 \sum_{(\zeta'\xi') \in \Omega_{\sigma'\eta'}^2}\mathsf{c}_{\sigma'\eta',\zeta'\xi'}~e_{\zeta'\xi'} \in \mathtt{F}_{\sigma\eta}.\]

    Therefore, for all $\sigma,\eta \in \Omega$, $\mathtt{F}_{\sigma\eta}$ is an ideal of $\mathcal{E}.$
\end{proof}

The next lemma establishes that all algebraic automorphisms of a Gibbs evolution algebra maps the standard base to itself. 

\begin{lemma} \label{lm:evolution.basis.automorphism}
    Let $\varphi: \mathcal{E} \to \mathcal{E}$ be an automorphism of the $\mathscr{C}$-evolution Gibbs algebra $\mathcal{E}=\mathcal{E}(\mathscr{C},\mu,\Phi,\Omega)$. Then $\varphi(\mathfrak{B}_{\Omega^2}) \subseteq \mathfrak{B}_{\Omega^2}$.
\end{lemma}
\begin{proof}
    First, observe that for every $\sigma \in \Omega$, $e_{\sigma\sigma}^2= e_{\sigma\sigma}$ and every idempotent element of $\mathcal{E}$ is given by $\sum_{\sigma' \in I'} e_{\sigma'\sigma'}$ with $I' \subseteq \Omega$ finite. Thus, there exists a finite $I \subseteq \Omega$ such that $\varphi(e_{\sigma\sigma}) = \sum_{\sigma' \in I} e_{\sigma'\sigma'}$. However, $e_{\sigma\sigma} = \sum_{\sigma' \in I}\varphi^{-1}(e_{\sigma'\sigma'})$ with $\varphi^{-1}(e_{\sigma'\sigma'}) = \sum_{\eta \in J(\sigma')}e_{\eta}$. Hence, $I=\{\sigma'\}$, \textit{i.e.}, $\varphi(e_{\sigma\sigma})= e_{\sigma'\sigma'}$.

    Now let $\varphi(e_{\sigma\eta}) = \sum_{{(\sigma',\eta')} \in \Xi_{\sigma\eta}}a_{\sigma'\eta'}e_{\sigma'\eta'}$ for any $\sigma, \eta \in \Omega$. Then,
    \[\sum_{(\sigma',\eta') \in \Xi_{\sigma\eta}}a_{\sigma'\eta'}^2 \sum_{(\zeta',\xi') \in \Omega_{\sigma'\eta'}^2} \mathsf{c}_{\sigma'\eta',\zeta'\xi'} ~e_{\zeta'\xi'} = \sum_{(\zeta,\xi) \in \Omega_{\sigma,\eta}^2}\mathsf{c}_{\sigma\eta,\zeta\xi} ~\varphi(e_{\zeta\xi}).\]

    By combining the result for idempotents with the equation above, we obtain that $a_{\sigma'\eta'}=1$ and $\Xi_{\sigma\eta} =\{(\sigma',\eta')\}$, and the lemma holds.
\end{proof}

The elements $e_{\sigma\eta}$ such that $e_{\sigma\eta}^2=0$ do not reproduce and, therefore, they do not act in the evolution dynamics of the system. Thus, it becomes interesting to consider an algebra in which all its elements are fertile. Set
\[\mathfrak{N}:=\left\{e_{\sigma\eta} \in \mathfrak{B}_{\Omega^2} : \mathfrak{D}_{\sigma\eta} \not\in \mathcal{L}\right\}.\]

Define $\mathcal{E}_M=\mathcal{E}_{M}(\mathscr{C},\mu,\Phi,\Omega)$ to be the subalgebra of 
$\mathcal{E}(\mathscr{C},\mu,\Phi,\Omega)$ such that $\mathcal{E}_M:=\big\langle \mathfrak{B}_{\Omega^2}\setminus\mathfrak{N} \big\rangle$. We call $\mathcal{E}_M$ the Markov $\mathscr{C}$\textit{-evolution Gibbs algebra} generated by the $\mu \in \mathscr{G}(\Phi)$ on $\Omega$.

The theorem below shows that $\mathcal{E}_M$ is a Markov evolution algebra and provides a decomposition into a direct sum of ideals that can be isomorphic. Before stating the result, let us fix some notation. Set $(\sigma,\eta)_L := (\sigma_L,\eta_L)$ for all $L \subseteq \mathbb{L}$ and $(\sigma,\eta) \in \Omega^2$. Let $(\sigma,\eta)_L(\zeta,\xi)_{L'}$ stand for $(\sigma_L\zeta_{L'},\eta_L\xi_{L'}) \in \big(S^{L \dot\cup L'}\big)^2$ when $L\cap L'=\emptyset$.

\begin{theorem}[Decomposition of $\mathcal{E}_M$ into a direct sum of ideals] \label{thm:isomorphic.decomposition.E.M}
    Let $\mathcal{E}_M := \mathcal{E}_M(\mathscr{C},\mu,\Phi,\Omega)$ be a Markov $\mathscr{C}$-evolution Gibbs algebra generated by $\mu \in \mathscr{G}(\Phi)$ on $\Omega$. Then $\mathcal{E}_M$ is indeed Markov  such that
    \[\mathcal{E}_M = \bigoplus_{\sigma \in \widetilde{\Omega}}\mathtt{F}_{\sigma\sigma},\]
    where each $\mathtt{F}_{\sigma\sigma} \in \mathcal{F}_{\Omega^2}$ is an ideal with countable basis $\mathfrak{B}_{(\mathtt{E}^\sigma)^2}$. Moreover, if $\Phi$ has finite range; then, for all $\sigma,\eta \in \Omega$, $\mathtt{F}_{\sigma\sigma}$ and $\mathtt{F}_{\eta\eta}$ are isomorphic.
\end{theorem}
\begin{proof}
    The decomposition into a the direct sum of ideals is immediately obtained by \cref{cor:evo.alg.decomposition}. Observe that, for all $e_{\sigma\eta}\in\mathfrak{B}_{\Omega^2}\setminus\mathfrak{N}$, $\sigma \in \mathtt{E}^\eta$. Hence, one can easily see that $\sum_{(\zeta,\xi) \in \Omega_{\sigma\eta}^2} \mathsf{c}_{\sigma\eta,\zeta\xi} =1$ for all $e_{\sigma\eta}\in\mathfrak{B}_{\Omega^2}\setminus\mathfrak{N}$ and $\mathcal{E}_M$ is Markov.

    Consider now $\Phi$ to be a potential with finite range. It remains to show that the ideals $\mathtt{F}_{\sigma\sigma}$ are isomorphic. Set $g_L^{\sigma\eta}:\Omega \to S^L$ to be such that, for all $x \in L$ and $\omega\in\Omega$,
    \[g_L^{\sigma\eta}(\omega)(x) = \left\lbrace\begin{array}{cc}
        \omega(x), & \text{if } \omega(x)\neq\eta(x) \\
        \sigma(x), & \text{otherwise},
    \end{array}\right.\]
    and write $g_L^{\sigma\eta}(\omega,\varpi)$ for $\big(g_L^{\sigma\eta}( \omega),g_L^{\sigma\eta}(\varpi)\big)$.
    
    Let us fix, without loss of generality, two distinct $\sigma, \eta \in \widetilde{\Omega}$. Define, for $\zeta,\xi \in \mathtt{E}^\sigma$,
    \[
        \Lambda_{\zeta\xi}^\Phi:= \mathfrak{D}_{\sigma\zeta} \setminus\operatorname{cl}_\Phi \left(\mathfrak{D}_{\zeta\xi}\right) \quad \text{and} \quad L_{\zeta\xi}^\Phi:=\left(\operatorname{cl}_\Phi\left(\mathfrak{D}_{\zeta\xi}\right) \cup \mathfrak{D}_{\sigma\zeta}\right)^c.
    \]
    
    Thus, $\mathbb{L}=\operatorname{cl}_\Phi \left(\mathfrak{D}_{\zeta\xi}\right) ~\dot\cup ~L_{\zeta\xi}^\Phi ~\dot\cup~ \Lambda_{\zeta\xi}^\Phi$ and $\operatorname{cl}_\Phi \left(\mathfrak{D}_{\zeta\xi}\right), \Lambda_{\zeta\xi}^\Phi \in \mathcal{L}$. Set $\varphi:\mathtt{F}_{\sigma\sigma}\to\mathtt{F}_{\eta\eta}$ to be given by linear extension of $\varphi(e_{\zeta\xi})=e_{f_{\sigma\eta}(\zeta,\xi)}$ such that \[f_{\sigma\eta}(\zeta,\xi) := (\zeta,\xi)_{\operatorname{cl}_\Phi \left(\mathfrak{D}_{\zeta\xi}\right)~} (\eta,\eta)_{L_{\zeta\xi}^\Phi~}g_{\Lambda_{\zeta\xi}^\Phi}^{\sigma\eta}(\zeta,\xi).\]
        
    Write $(\zeta',\xi')$ for $f_{\sigma\eta}(\zeta, \xi)= (\zeta',\xi')$. Since $\operatorname{cl}_\Phi \left(\mathfrak{D}_{\zeta\xi}\right)\cup \Lambda_{\zeta\xi}^\Phi \in \mathcal{L}$, one can easily see that $(\zeta',\xi') \in (\mathtt{E}^\eta)^2$. Observe that $\zeta_{\Lambda_{\zeta\xi}^\Phi} = \xi_{\Lambda_{\zeta\xi}^\Phi}$ and thence $\mathfrak{D}_{\zeta\xi} = \mathfrak{D}_{\zeta'\xi'}$. We also have
    \[\Lambda_{\zeta\xi}^\Phi = \mathfrak{D}_{\eta\zeta'}\setminus\operatorname{cl}_\Phi\left(\mathfrak{D}_{\zeta'\xi'}\right)= \mathfrak{D}_{\eta\xi'}\setminus\operatorname{cl}_\Phi\left(\mathfrak{D}_{\zeta'\xi'}\right) = \Lambda_{\xi\zeta}^\Phi.\]
    \begin{figure}[!htb]
        \centering
        \includegraphics[width=300pt]{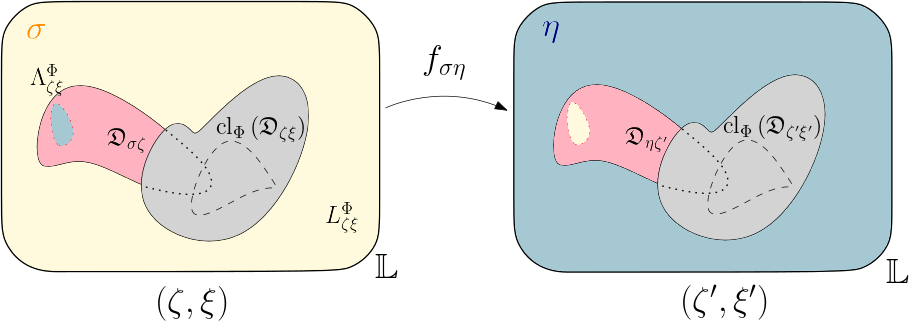}
        \caption{Graphical representation of $f_{\sigma \eta}$ mapping configurations of $(\mathtt{E}^\sigma)^2$ to $(\mathtt{E}^\eta)^2$.}
        \label{fig:evo.iso.fertile}
    \end{figure}

    Let now $\omega,\varpi \in \mathtt{E}^\sigma$ be such that $f_{\sigma\eta}(\omega,\varpi)= (\zeta',\xi')$. Then,
    \[(\omega,\varpi)_{\operatorname{cl}_\Phi\left(\mathfrak{D}_{\zeta\xi}\right)} = (\zeta,\xi) _{\operatorname{cl}_\Phi\left(\mathfrak{D}_{\zeta\xi}\right)}, \quad\]
    \[(\omega,\varpi) _{\mathfrak{D}_{\sigma\zeta} \setminus \operatorname{cl}_\Phi\left(\mathfrak{D}_{\zeta\xi}\right)} = g_{\mathfrak{D}_{\eta\zeta'}\setminus\operatorname{cl}_\Phi\left(\mathfrak{D}_{\zeta'\xi'}\right)}^{\eta\sigma}(\zeta',\xi')=(\zeta,\xi)_{\Lambda^\Phi_{\zeta\xi}}, \quad\]
    and $(\omega,\varpi)_{L_{\zeta\sigma}^\Phi} = (\sigma,\sigma)_{L_{\zeta\sigma}^\Phi}$. It implies that 
 \begin{align*}
     &\omega = \omega_{\mathfrak{D}_{\sigma\zeta}}\sigma_{(\mathfrak{D}_{\sigma\zeta})^c}= \zeta_{\mathfrak{D}_{\sigma\zeta}}\sigma_{(\mathfrak{D}_{\sigma\zeta})^c}= \zeta, \\
     \text{and}\quad &\varpi = \varpi_{\mathfrak{D}_{\sigma\zeta}}\sigma_{(\mathfrak{D}_{\sigma\zeta})^c}= \xi_{\mathfrak{D}_{\sigma\zeta}}\sigma_{(\mathfrak{D}_{\sigma\zeta})^c}= \xi.
 \end{align*}
 Hence, $f_{\sigma\eta}$ is injective. Note that, for any given $\omega', \varpi' \in \mathtt{E}^\eta$, \[\thickbar{\Lambda}_{\omega'\varpi'}^\Phi := \mathfrak{D}_{\eta\omega'}\setminus\operatorname{cl}_\Phi\left(\mathfrak{D}_{\omega'\varpi'}\right), \quad\text{and}\quad \thickbar{L}_{\omega'\varpi'}^\Phi:=\left(\operatorname{cl}_\Phi\left(\mathfrak{D}_{\omega'\varpi'}\right) \cup \mathfrak{D}_{\eta\omega'}\right)^c,\]
 one has that
 \[(\omega,\varpi):= (\omega',\varpi')_{\operatorname{cl}_\Phi \left(\mathfrak{D}_{\omega'\varpi'}\right)~} (\sigma,\sigma)_{\thickbar{L}_{\omega'\varpi'}^\Phi~}g_{\thickbar{\Lambda}_{\omega'\varpi'}^\Phi}^{\eta\sigma}(\omega',\varpi') \in \mathtt{E}^\sigma\]
 is such that $f_{\sigma\eta}(\omega,\varpi)=(\omega',\varpi')$. Therefore, $f_{\sigma\eta}$ and $\varphi$ are bijective. We will now prove that $\varphi$ is an isomorphism of ideals by showing that $\varphi(e_{\zeta\xi}^2)= e_{\zeta'\xi'}^2$. Observe that since $\mathfrak{D}_{\zeta\xi}= \mathfrak{D}_{\zeta'\xi'}$ and $(\zeta,\xi)_{\mathfrak{D}_{\zeta\xi}} = (\zeta',\xi')_{\mathfrak{D}_{\zeta'\xi'}}$, it follows that  if $(\omega,\varpi) \in \Omega_{\zeta\xi}^2$, then  $ (\omega',\varpi'):=f_{\sigma\eta}(\omega,\varpi) \in \Omega_{\zeta'\xi'}^2 = f_{\sigma\eta}(\Omega_{\zeta\xi}^2)$ is such that 
 \begin{equation*} \label{eq:omega.local.comparison}
    \omega_{\operatorname{cl}_\Phi(\mathfrak{D}_{\zeta\xi})}=\omega'_{\operatorname{cl}_\Phi(\mathfrak{D}_{\zeta'\xi'})} \quad \text{and} \quad \varpi_{\operatorname{cl}_\Phi(\mathfrak{D}_{\zeta\xi})}=\varpi'_{\operatorname{cl}_\Phi(\mathfrak{D}_{\zeta'\xi'})}.
 \end{equation*}
 
 By the equations above and by the $\mathscr{F}_A$-mensurability of $\Phi_A$, one has that, for all $\omega \in \Omega_{\zeta\xi}$, 
 \begin{align*}
    H_{\mathfrak{D}_{\zeta\xi}}^\Phi(\omega)&= \sum_{A \subseteq \operatorname{cl}_\Phi(\mathfrak{D}_{\zeta\xi}), ~A \cap \mathfrak{D}_{\zeta\xi}\neq \emptyset}\Phi_A(\omega)\\
    &= \sum_{A \subseteq \operatorname{cl}_\Phi(\mathfrak{D}_{\zeta'\xi'}), ~A \cap \mathfrak{D}_{\zeta'\xi'} \neq \emptyset}\Phi_A(\omega') =H_{\mathfrak{D}_{\zeta'\xi'}}^\Phi(\omega').
 \end{align*}

 Therefore, $\mathsf{c}_{\zeta\xi,\omega\varpi} = \mathsf{c}_{\zeta'\xi',\omega'\varpi'}$ for all $\zeta, \xi \in \mathtt{E}^\sigma$ and every $\omega,\varpi \in \Omega_{\zeta\xi}$, an it completes the proof.
\end{proof}

We present below an example of two fertile ideals of $\mathcal{E}_M$ that are not isomorphic when the interaction potential has infinite range.

\begin{example}(Fertile Markov ideals are not necessarily isomorphic)
    Let us consider $\mathbb{N}=\{1, 2, \dots\}$ and fix $\mathbb{L}=\mathbb{N}_0:=\mathbb{N} \cup\{0\}$. Set $S=\{0,1\}$ to be the spin set and $\mathscr{C}_{\blacktriangle}= \{\mathbb{N}_0\}$ the partition with an unique cluster. Consider $\Phi$ such that
    \[\Phi_A(\sigma):= \left\{\begin{array}{cl}
         \dfrac{1}{n^2}\sigma(n)\sigma(0),& \text{if } A = \{0, n\} \text{ with } n \in \mathbb{N}; \\
         0, & \text{otherwise}.
    \end{array}\right.\]
    Thus, $\Phi_A$ is $\mathscr{F}_A$-measurable and
    \[H_\Lambda^\Phi(\sigma)= \left\{\begin{array}{rl}
        \sigma(0)\sum\limits_{n \in \Lambda} \dfrac{1}{n^2} \sigma(n), & \text{if } \Lambda \cap \{0\} = \emptyset;\\
        \sigma(0)\sum\limits_{n \in \mathbb{N}} \dfrac{1}{n^2}\sigma(n), & \text{ otherwise}.
    \end{array} \right.\]

    Note that, for all $\Lambda \in \mathcal{L}$ and $\sigma \in \Omega$, one has that $0 \leq H_\Lambda^\Phi(\sigma) < 2$. Therefore $\Phi$ is an admissible potential and there exists a Gibbs measure $\mu \in \mathscr{G}(\Phi)$. One can easily see that $\Phi$ does not have finite range because $\mathbf{n}_0^\Phi= \mathbb{N}$ and $\mathcal{G}_\Phi$ is an infinite star graph.
    
    Let us fix $\zeta, \xi' \in \Omega$ such that
    \[\zeta \equiv 0 \quad \text{and} \quad \xi' \equiv 1.\]

    We will show that $\mathtt{F}_{\zeta\zeta}$ is not isomorphic to $\mathtt{F}_{\xi'\xi'}$ as subalgebras of $\mathcal{E}_M(\mathscr{C}_{\blacktriangle},\mu,\Phi,\Omega)$. Set $\eta' := \zeta_{\{0\}}\xi'_{\mathbb{N}}$. Then $\Omega_{\xi'\eta'} = \{\xi',\eta'\}$ and
    \[c_{\xi'\eta',\xi'} = \frac{1}{1+ h_{\{0\}}^\Phi(\eta')/h_{\{0\}}^\Phi(\xi')} = \frac{1}{1+ \exp(\sum_{n = 1}^{+\infty}{1}/{n^2})}.\]

    Suppose, by contradiction, that there exist $\xi,\eta \in \mathtt{E}^{\zeta}$ such that $c_{\xi\eta, \sigma} = c_{\xi'\eta', \xi'}$ with $\sigma \in \Omega_{\xi\eta}$. Since $\mathscr{C}_{\blacktriangle}$ has an unique cluster, $\sigma \in\{\xi, \eta\}$. Hence, it suffices to find $\xi, \eta \in \mathtt{E}^\zeta$ such that $H_{\mathfrak{D}_{\xi\eta}}(\xi)- H_{\mathfrak{D}_{\xi\eta}}(\eta) = \pm\sum_{n=1}^{+\infty} \frac{1}{n^2}$. However, \[\xi = \xi'_{\Lambda'}\zeta_{(\Lambda')^c}\text{ and }\eta = \xi'_{\Lambda''}\zeta_{(\Lambda'')^c} \quad \text{ with } \Lambda',\Lambda'' \in \mathcal{L}.\]
    Thus, $\mathfrak{D}_{\xi\eta} = \Lambda' \triangle \Lambda''$. If $0 \in \mathfrak{D}_{\xi\eta}$, then
    \[\big\vert H_{\mathfrak{D}_{\xi\eta}}(\xi)- H_{\mathfrak{D}_{\xi\eta}}(\eta)\big\vert = \left\vert\sum_{n \in \Lambda' \setminus \Lambda''} \frac{1}{n^2}\xi(0) - \sum_{m \in \Lambda''\setminus\Lambda'} \frac{1}{m^2}\eta(0)\right\vert <\sum_{n=1}^{+\infty} \frac{1}{n^2}.\]
    Now, if $0 \not\in \mathfrak{D}_{\xi\eta}$, then \[H_{\mathfrak{D}_{\xi\eta}}(\xi)- H_{\mathfrak{D}_{\xi\eta}}(\eta) = \big(\xi(0)-\eta(0)\big)\sum_{n=1}^{+\infty} \frac{1}{n^2} =0.\]
    
    Therefore, for all $\xi,\eta \in \mathtt{E}^{\zeta}$ and every $\sigma \in \Omega_{\xi\eta}$, $c_{\xi'\eta', \xi'} < c_{\xi\eta, \sigma}$. It implies that
    \[\mathsf{c}_{\xi'\eta',\xi'\xi'} \neq c_{\xi\eta, \sigma}c_{\xi\eta, \omega}\] for all $\xi,\eta \in \mathtt{E}^\zeta$ and all $\sigma,\omega \in \Omega_{\xi\eta}$. We conclude that  $\mathtt{F}_{\zeta\zeta}$ and  $\mathtt{F}_{\xi'\xi'}$ are not isomorphic. \hfill $\blacktriangleleft$ 
\end{example}

The following propositions are direct consequences of theorems from Section \ref{sec:infinite.dim.genetic.algebras}. We chose to write the results for the Markov evolution algebra $\mathcal{E}_M$, but they are also true for $\mathcal{E}$. Before stating the propositions, we define that a Markov $\mathscr{C}$-evolution algebra $\mathcal{E}_m$ is \textit{$\tau$-isomorphic to} $\mathcal{E}_M'$ when the linear map such that $\phi(e_{\sigma\eta})=e_{\tau\sigma\tau\eta}'$ determines an isomorphism of algebras. 

\begin{proposition} \label{prop:evo.T.equiv}
    Let $\tau \in \mathcal{T}$ and let $\mathcal{E}_M = \mathcal{E}_M(\mathscr{C},\mu,\Phi,\Omega)$ and $\mathcal{E}_M'= \mathcal{E}_M\big(\tau(\mathscr{C}),\mu',\Psi,\Omega\big)$ be two evolution Gibbs algebras. If $\Phi\sim \tau^{-1}(\Psi)$, then the algebra $\mathcal{E}_M$ is $\tau$-isomorphic to $\mathcal{E}_M'$. Moreover, the converse holds when $\mathscr{C}=\mathscr{C}_{\odot}$ is the set of atomic clusters. 
\end{proposition}
\begin{proof}
    The proof follows repeating exactly the same arguments of \cref{thm:genetic.T.equiv.potentials}. The first part is straightforward from the definition of the structure coefficients of $\mathcal{E}_M$. For the second half we first obtain that there exists $k_{\zeta\eta}^2>0$ such that, for all $\sigma,\omega \in \Omega_{\zeta\eta}$,
    \begin{equation*} 
        k_{\zeta\eta}^2:=\frac{h_{\mathfrak{D_{\zeta \eta}}}^\Phi(\sigma)h_{\mathfrak{D_{\zeta \eta}}}^\Phi(\omega)}{h_{\mathfrak{D_{\tau\zeta \tau\eta}}}^{\Psi}(\tau\sigma)h_{\mathfrak{D_{\tau\zeta \tau\eta}}}^{\Psi}(\tau\omega)}=\left(\frac{\sum\limits_{\xi \in \Omega_{\zeta\eta}}h_{\mathfrak{D_{\zeta \eta}}}^\Phi(\xi)}{\sum\limits_{\xi' \in \Omega_{\tau\zeta\tau\eta}}h_{\mathfrak{D_{\tau\zeta \tau\eta}}}^{\Psi}(\xi')}\right)^2.
    \end{equation*}
    Thus we recover $k_{\zeta\eta}$ defined in \eqref{eq:k.offspring} by the equation above and the proposition holds.
\end{proof}

\begin{corollary}[Stability under phase transition] \label{cor:evo.stability.phase.transition}
   Let $\mu, \mu' \in \mathscr{G}(\Phi)$ be Gibbs measures on $\Omega$. Then the algebra $\mathcal{E}_M(\mathscr{C},\mu,\Phi,\Omega)$ is isomorphic to $\mathcal{E}_M(\mathscr{C},\mu',\Phi,\Omega)$.
\end{corollary}

\begin{proposition}[Evolution algebras generated by products of Gibbs measures]
    Let $\{\mathbb{L}_i\}_{i=1}^n$ be a sequence of countable sets such that, for each $i\in \{1, \dots, n\}$, $\mathscr{C}_i$ is a partition of $\mathbb{L}_i$ associated with a Gibbs measure $\mu_i \in \mathscr{G}(\Phi^i)$ on $\Omega_i = S^{\mathbb{L}_i}$ with $S$ a fixed finite set of spins.
    
    Consider the evolution Gibbs algebras $\mathcal{E}_{M,i}:= \mathcal{E}_{M}(\mathscr{C}_i,\mu_i,\Phi^i,\Omega_i)$ for  all $i \in \{1, \dots, n\}$,  and \[\mathcal{E}_M=\mathcal{E}_M\left(\bigsqcup_{i=1}^n \mathscr{C}_i,\bigotimes_{i=1}^n\mu_i,\bigoplus_{i=1}^n\Phi^i,\prod_{i=1}^n\Omega_i\right).\] Then $\mathcal{E}$ is isomorphic to the tensor algebra $\bigotimes_{i =1}^n \mathcal{E}_{M,i}$ equipped with the ordinary product.
\end{proposition}
\begin{proof}
    We verify this results following similar steps as in the proof of \cref{thm:genetic.product.v.tensor.alg}. Let us recall the same notation defined in \cref{thm:genetic.product.v.tensor.alg}. 
    
    By \eqref{eq:equality.product.measure} and by the commutativity and associativity of the product measure, for all $\sigma, \xi \in \Omega$, each $\Lambda \in \mathcal{L}$
    \[\mu^{\otimes 2}(\sigma\mid\xi_{\Lambda^c})= \left(\bigotimes_{i=1}^n \mu_i^{\otimes 2}\right) (\sigma\mid\xi_{\Lambda^c}).\]

    Let $\sigma,\eta,\sigma',\eta' \in \Omega$ and write $\sigma= \sigma_1 \cdots \sigma_n$ and $\eta= \eta_1 \cdots \eta_n$ with $\sigma_i,\eta_i\in \Omega_i$ for all $i \in \{1, \dots, n\}$. Consider $\varphi: \mathcal{E}_M \to \bigotimes_{i=1}^n\mathcal{E}_{M,i}$ to be given by linear extension of \[\varphi(e_{\sigma\eta}) = e_{\sigma_1\eta_1} \otimes \cdots \otimes e_{\sigma_n \eta_n}.\]

    Now, one can easily see that
    \[
        \varphi(e_{\sigma\eta} \cdot e_{\sigma'\eta'})= \left\lbrace \begin{array}{cl}
            \sum\limits_{(\zeta,\xi) \in \Omega_{\sigma\eta}^2} \prod_{i=1}^n\mathsf{c}_{\sigma_i \eta_i, \zeta_i\xi_i} ~\varphi(e_{\zeta\xi}), & \text{ if } \sigma = \sigma' \text{ and } \eta=\eta';\\
        0, & \text{otherwise}. 
        \end{array} \right.
    \]

    Therefore, $\varphi(e_{\sigma\eta} \cdot e_{\sigma'\eta'}) = \varphi(e_{\sigma\eta}) \cdot \varphi(e_{\sigma'\eta'})$ for all $(\sigma,\eta), (\sigma',\eta') \in \Omega^2$ which implies that $\varphi$ is an isomorphism of algebras.
\end{proof}

The next proposition is similar to \cref{thm:finite.subalgebras}. Since \cref{thm:isomorphic.decomposition.E.M} ensures us that all fertile ideals are isomorphic when $\Phi$ has finite range, we state the next theorem only for one fertile ideal $\mathtt{F}_{\sigma\sigma}$.

\begin{proposition}[Finite-dimensional subalgebras of fertile ideals] \label{prop:finite.subalgebra}
    Consider $\mathtt{F}_{\sigma\sigma}$ to be a fertile ideal of $\mathcal{E}_M(\mathscr{C}, \mu,\Phi, \Omega)$ such that the interaction potential $\Phi$ has finite range. If $\mathtt{A}$ is a finite-dimensional subalgebra of $\mathtt{F}_{\sigma\sigma}$, then there exists $\Lambda \in \mathcal{L}$ such that $\mathtt{A}$ is isomorphic to a subalgebra of $\mathlcal{E}_f\big(\mathscr{C}\vert_\Lambda,\mu_\Lambda(\cdot\mid\xi),S^\Lambda\big)$ with $\xi \in \Omega$ arbitrary.

    The same holds for $\mathlcal{E}_f\big(\mathscr{C}\vert_{\Lambda'},\mu_{\Lambda'}(\cdot\mid\xi),S^{\Lambda'}\big)$ with $\Lambda' \in \mathcal{L}$ such that $\Lambda \subseteq \Lambda'$.
\end{proposition}
\begin{proof}
    Let $\mathtt{E}_f$ be the finite subset of $(\mathtt{E}^\sigma)^2$ such that $\mathfrak{B}_{\mathtt{E}_f}$ is a basis of $\mathtt{A}$. Fix
    \[\Lambda:= \bigcup_{(\zeta,\eta) \in \mathtt{E}_f} \operatorname{cl}_\Phi(\mathfrak{D}_{\zeta\eta})\]and let $\Lambda' \in \mathcal{L}$ be arbitrary such that $\Lambda \subseteq \Lambda'$. It suffices to show the result for $\mathlcal{E}_f := \mathlcal{E}_f\big(\mathscr{C}\vert_{\Lambda'},\mu_{\Lambda'}(\cdot\mid\xi),S^{\Lambda'}\big)$. Fix $f: \mathtt{E}_f \to \big(S^{\Lambda'}\big)^2$ to be given by $f(\zeta,\eta)= (\zeta,\eta)_{\Lambda}(\xi,\xi)_{\Lambda'\setminus\Lambda}$. Set $\varphi: \mathtt{A}\to \mathlcal{E}_f$ to be given by linear extension of
    \[
        \varphi(e_{\zeta\eta})= e_{f(\zeta,\eta)}
    \]

    We now proceed applying the same steps as in the proof of \cref{thm:finite.subalgebras} to verify that $\varphi$ is multiplicative. Hence, by the algebraic properties inherited from $\mathtt{A}$, $\operatorname{Im}(\varphi)$ is a subalgebra of $\mathlcal{E}_f$.
\end{proof}

\begin{remark*}
\cref{prop:evo.T.equiv,prop:finite.subalgebra} may seem more general than they are when we compare them with Theorem 3.3 of Rozikov and Tian \cite{rozikov2011}. However, they showed in \cite{rozikov2011} that the finite algebras generated by distinct measures $\mu$ and $\mu'$ are isomorphic as vector spaces, but not as algebras. The linear map used in their proof is not multiplicative. In fact, if we apply \cref{lm:finite.evo.alg.equiv,lm:evolution.basis.automorphism}, we verify that it is not always possible to have an automorphic change of basis of  $\mathlcal{E}_f(\mathscr{C},\mu,\Omega)$ that equates its structure coefficients to those of $\mathlcal{E}_f(\mathscr{C},\mu',\Omega)$ for an arbitrary Gibbs measure $\mu'$ on $\Omega$.  
\end{remark*}


\section{Final remarks and open questions}

In this article we managed to explore algebraic structures of genetic and evolution algebras whose structure constants are given in terms of configurations associated with a given Gibbs measure. We propose a definition that makes it possible to relate the infinite-dimensional algebras with the finite-dimensional one present in the literature. A valuable contribution of this work is that we successfully address a question made in \cite{rozikov2011} extending their construction to the infinite-dimensional case, dealing with the concerns highlighted in the introduction. However, some open problems remain.

We would like to pose some questions that may guide further studies on this topic:
\begin{enumerate}[(a)]
    \item Is it possible to modify the algebras preserving part of their properties to identify the phase transition phenomenon? We obtained in \cref{cor:stability.phase.transition,cor:evo.stability.phase.transition} that the algebras generated by measures in $\mathscr{G}(\Phi)$ are isomorphic. It is a consequence of the Hamel basis and infertility between distinct fertile classes.
    
    \item Techniques from functional analysis could be interesting to study more properties of the algebras. How do the algebras change when consider a Schauder basis for the fertile ideals? (recall that, by \cref{thm:genetic.decomposition,cor:evo.alg.decomposition}, the fertile ideals have countable basis). Indeed, in \cite {vidal2022hilbert} Vidal et al. considered such basis in order to provide an evolution algebra structure to a given Hilbert space.
    
    This question could be extended to other types of basis. \cref{thm:hom.A} may be useful for studies in that direction, but its result depends on the Hamel basis.

    \item Is it suitable to  apply consolidated techniques used for Gibbs measures (e.g., thermodynamic limit, perfect simulation, etc) to study genetic and evolution algebras?

    \item \cref{thm:genetic.T.equiv.potentials,prop:evo.T.equiv} establish a relation between equivalent potentials and isomorphic algebras. Is it possible to improve this result and exhibit characterization of the isomorphism?
    
    \item How to define similar algebras when $S$ infinite? Several properties of Gibbs measures applied in this article rely on the finiteness of $S$. One of the key aspects is that some consequences of Hammersley-Clifford Theorem are no longer valid. One difficulty is that we can not guarantee $c_{\zeta\eta,\sigma}>0$ when $\vert S\vert=\infty$. We believe that some progress may be achieved by studying infinite volume Gibbs measures that can be obtained as invariant measures of Gibbs sampler using a perfect simulation approach.
\end{enumerate}

\section{Declaration}
All authors declare that they have no conflicts of interest.

\section{Acknowledgements}
We sincerely thank the reviewers for their valuable comments and suggestions, which have helped improve the clarity and quality of this manuscript. This study was financed in part by the Coordenação de Aperfeiçoamento de Pessoal de Nível Superior - Brasil (CAPES) - Finance Code 001. It was also supported by grants \#2017/10555-0 and \#2019/19056-2 S\~ao Paulo Research Foundation (FAPESP).

\small
\bibliographystyle{abbrvnat}
\bibliography{references}

\end{document}